\documentclass[fleqn,12pt,twoside]{article}
\usepackage{tipa}

\usepackage[headings]{espcrc1}
\readRCS
$Id: espcrc1.tex,v 1.2 2004/02/24 11:22:11 spepping Exp $
\usepackage{graphicx}
\usepackage[figuresright]{rotating}

\newtheorem{theorem}{Theorem}

\newtheorem{claim}{Claim}

\newtheorem{conjecture}{Conjecture}
\newtheorem{corollary}{Corollary}

\newtheorem{lemma}{Lemma}

\newtheorem{proposition}{Proposition}
\newtheorem{remark}{Remark}

\newenvironment{proof}[1][Proof.]{\begin{trivlist}
\item[\hskip \labelsep {\bfseries #1}]}{\end{trivlist}}

\newenvironment{acknowledgement}[1][Acknowledgement]{\begin{trivlist}
\item[\hskip \labelsep {\bfseries #1}]}{\end{trivlist}}

\newcommand{\AmS}{{\protect\the\textfont2
  A\kern-.1667em\lower.5ex\hbox{M}\kern-.125emS}}

\usepackage{amsmath}
\usepackage{amsfonts}
\title{On disjoint matchings in cubic graphs}

\author{Vahan V. Mkrtchyan\address[MCSD]{Department of Informatics and Applied Mathematics,\\
Yerevan State University, Yerevan, 0025, Armenia}%
\address{Institute for Informatics and Automation Problems,\\
National Academy of Sciences of Republic of Armenia, 0014, Armenia}
\thanks{The author is supported by a grant of Armenian National Science and
Education Fund}
\thanks{email: vahanmkrtchyan2002@\{ysu.am, ipia.sci.am,
yahoo.com\}},
        Samvel S. Petrosyan\addressmark[MCSD]\thanks{email: samvelpetrosyan2008@yahoo.com
},
                and
        Gagik N. Vardanyan\addressmark[MCSD]\thanks{email: vgagik@gmail.com.}}

\runtitle{On disjoint matchings in cubic graphs} \runauthor{Vahan
Mkrtchyan, Samvel Petrosyan, Gagik Vardanyan}

\begin{document}

\maketitle

\begin{abstract}
For $i=2,3$ and a cubic graph $G$ let $\nu _{i}(G)$ denote the
maximum number of edges that can be covered by $i$ matchings. We
show that $\nu _{2}(G)\geq \frac{4}{5}\left\vert V(G)\right\vert $
and $\nu _{3}(G)\geq \frac{7}{6}\left\vert V(G)\right\vert $.
Moreover, it turns out that $\nu _{2}(G)\leq \frac{|V(G)|+2\nu
_{3}(G)}{4}$.
\end{abstract}

\section{Introduction}

In this paper graphs are assumed to be finite, undirected and without
loops, though they may contain multiple edges. We will also consider
pseudo-graphs, which, in contrast with graphs, may contain loops.
Thus graphs are pseudo-graphs. We accept the convention that a loop
contributes to the degree of a vertex by two.

The set of vertices and edges of a pseudo-graph $G$ will be denoted
by $V(G)$ and $E(G)$, respectively. We also define: $n=|V(G)|$ and
$m=|E(G)|$. We will also use the following scheme for notations:
if $G$ is a pseudo-graph and $f$ is a graph-theoretic
parameter, we will write just $f$ instead of $f(G)$. So, for
example, if we would like to deal with the edge-set of a
pseudo-graph $G^{(0)*}_{i}$, we will write $E^{(0)*}_{i}$ instead of
$E(G^{(0)*}_{i})$; moreover we will write $m^{(0)*}_{i}$ for the
number of edges in this graph.

A connected $2$-regular graph with at least two vertices will be
called a \textit{cycle}. Thus, a loop is not considered to be a
cycle in a pseudo-graph. Note that our notion of cycle differs from
the cycles that people working on nowhere-zero flows and cycle
double covers are used to deal with.

The length of a path or a cycle is the number of edges lying on it.
The path or cycle is even (odd) if its length is even (odd). Thus,
an isolated vertex is a path of length zero, and it is an even path.

For a graph $G$ let $\Delta=\Delta (G)$ and $\delta=\delta (G)$
denote the maximum and minimum degrees of vertices in $G$,
respectively. Let $\chi ^{\prime}=\chi ^{\prime }(G)$ denote the
chromatic class of the graph $G$.

The classical theorem of Shannon states:

\begin{theorem}
(Shannon \cite{Shannon}). For every graph $G$%
\begin{equation}
\Delta\leq \chi ^{\prime }\leq \left[ \frac{3\Delta}{2}\right].
\end{equation}
\end{theorem}

In 1965 Vizing proved:

\begin{theorem}
(Vizing, \cite{Vizing}): $\Delta\leq \chi ^{\prime }\leq \Delta +\mu
$, where $\mu$ denotes the maximum multiplicity of an edge in $G$.
\end{theorem}

Note that Shannon's theorem implies that if we consider a cubic
graph $G$, then $3\leq \chi ^{\prime }\leq 4$, thus $\chi ^{\prime
}$ can take
only two values. In 1981 Holyer proved that the problem of deciding whether $%
\chi ^{\prime }=3$ or not for cubic graphs $G$ is NP-complete \cite%
{Holyer}, thus the calculation of $\chi ^{\prime }$ is already hard
for cubic graphs.

For a graph $G$ and a positive integer $k$ define

\begin{center}
$B_{k}\equiv \{(H_{1},...,H_{k}):H_{1},...,H_{k}$ are pairwise
edge-disjoint matchings of $G\}$,
\end{center}

and let

\begin{center}
$\nu _{k}\equiv \max \{\left\vert H_{1}\right\vert +...+\left\vert
H_{k}\right\vert :(H_{1},...,H_{k})\in B_{k}\}$.
\end{center}

Define:

\begin{center}
$\alpha _{k}\equiv \max \{\left\vert H_{1}\right\vert
,...,\left\vert H_{k}\right\vert :$ $(H_{1},...,H_{k})\in B_{k}$ and
$\left\vert H_{1}\right\vert +...+\left\vert H_{k}\right\vert =\nu
_{k}\}$.
\end{center}

If $\nu $ denotes the cardinality of the largest matching of $G$,
then it
is clear that $\alpha _{k}\leq \nu $ for all $G$ and $k$. Moreover, $%
\nu _{k}=\left\vert E\right\vert =m$ for all $k\geq \chi ^{\prime
}$. Let us also note that $\nu _{1}$ and $\alpha _{1}$ coincide with
$\nu $.

In contrast with the theory of $2$-matchings, where every graph $G$
admits a maximum $2$-matching that includes a maximum matching
\cite{Lov}, there are graphs that do not have a \textquotedblleft
maximum\textquotedblright\ pair of disjoint matchings (a pair
$(H,H^{\prime })\in B_{2}$ with $\left\vert H\right\vert +\left\vert
H^{\prime }\right\vert =\nu _{2}$) that includes a maximum matching.

The following is the best result that can be stated about the ratio
$\nu
/\alpha _{2}$ for any graph $G$ (see \cite{FiveFourth}):%
\begin{equation}
1\leq \nu /\alpha _{2}\leq 5/4.
\end{equation}

Very deep characterization of graphs $G$ satisfying $\nu /\alpha
_{2}=5/4$ is given in \cite{FivefourthCharacter}.

Let us also note that by Mkrtchyan's result \cite{MPP0-1},
reformulated as in \cite{HararyPlummer}, if $G$ is a matching
covered tree, then $\alpha
_{2}=\nu $. Note that a graph is said to be matching covered (see \cite%
{Perfect}), if its every edge belongs to a maximum matching (not
necessarily a perfect matching as it is usually defined, see e.g.
\cite{Lov}).

The basic problem that we are interested is the following: what is
the proportion of edges of an $r$-regular graph (particularly, cubic
graph), that we can cover by its $k$ matchings? The formulation of
our problem stems from the recent paper \cite{KaiserKralNorine},
where the authors investigate the proportion of edges of a
bridgeless cubic graph that can be covered by $k$ of its perfect
matchings.

The aim of the present paper is the investigation of the ratios $\nu
_{k}/\left\vert E\right\vert $ (or equivalently, $\nu
_{k}/\left\vert V\right\vert $) in the class of cubic graphs for $%
k=2,3 $. Note that for cubic graphs $G$ Shannon's theorem implies
that $\nu _{k}=\left\vert E\right\vert ,$ $k\geq 4$.

The case $k=1$ has attracted much attention in the literature. See \cite%
{Hobbs} for the investigation of the ratio in the class of simple
cubic graphs, and
\cite{Bella,HenningYeo,TakaoBaybars,Takao,Weinstein} for the
general case. Let us also note that the relation between $\nu _{1}$ and $%
\left\vert V\right\vert $ has also been investigated in the regular
graphs of high girth \cite{GirthBound}.

The same is true for the case $k=2,3$. Albertson and Haas
investigate these ratios in the class of simple cubic graphs (i.e.
graphs without multiple edges)in \cite
{AlbertsonHaasFirst,AlbertsonHaasSecond}, and Steffen investigates
the general case in \cite{Steffen}.

\section{Some auxiliary results}

If $G$ is a pseudo-graph, and $e=(u,v)$ is an edge of $G$, then $k$%
-subdivision of the edge $e$ results a new pseudo-graph $G'$ which
is obtained from $G$ by replacing the edge $e$ with a path $P_{k+1}$
of length $k+1$, for which $V(P_{k+1})\cap V=\{u,v\}$. Usually, we
will say that $G'$ is obtained from $G$ by $k$-subdividing the edge
$e$.

If $Q$ is a path or cycle of a pseudo-graph $G$, and the pseudo-graph $%
G'$ is obtained from $G$ by $k$-subdividing the edge $e$, then
sometimes we will speak about the path or cycle $Q'$ corresponding
to $Q$, which roughly, can be defined as $Q$, if $e$ does not lie on
$Q$, and
the path or cycle obtained from $Q$ by replacing its edge $e$ with the path $%
P_{k+1}$, if $e$ lies on $Q$.

Our interest towards subdivisions is motivated by the following

\begin{proposition}
\label{CubicPseudoGraph}Let $G$ be a connected graph with $2\leq
\delta
\leq \Delta =3$. Then, there exists a connected cubic pseudo-graph $%
G_{0}$ and a mapping $k:E_{0}\rightarrow Z^{+}$, such that $G$ is
obtained from $G_{0}$ by $k(e)$-subdividing each edge $e\in E_{0}$,
where $Z^{+}$ is the set of non-negative integers.
\end{proposition}
\begin{proof}The existence of such a cubic pseudo-graph $G_{0}$ can be verified,
for example, as follows; as the vertex-set of $G_0$, we take the set
of vertices of $G$ having degree three, and connect two vertices
$u,v$ of $G_0$ by an edge $e=(u,v)$, if these vertices are connected
by a path $P$ of length $k,k\geq1$ in $G$, whose end-vertices are
$u$ and $v$, and whose internal vertices are of degree two. We also
define $k(e)=k-1$. Finally, if a vertex $w$ of $G_0$ lies on a cycle
$C$ of length $l,l\geq1$ in $G$, whose all vertices, except $w$, are
of degree two, then in $G_0$ we add a loop $f$ incident to $w$, and
define $k(f)=l-1$. Now, it is not hard to verify, that $G_0$ is a
cubic pseudo-graph, and if we $k(d)$-subdivide each edge $d$ of
$G_0$, then the resulting graph is isomorphic to $G$.
\end{proof}

Let $G_{0}$ be a cubic pseudo-graph, and let $e$ be a loop of $G_{0}$. Let $%
f $ be the edge of $G_{0}$ adjacent to $e$ (note that $f$ is not a
loop). Let $u_{0}$ be\ the vertex of $G_{0}$ that is incident to $f$
and $e$, and
let $f=(u_{0},v_{0})$. Assume that $v_{0}$ is not incident to a loop of $%
G_{0}$, and let $h$ and $h'$ be the other ($\neq f$) edges of $%
G_{0} $ incident to $v_{0}$, and assume $u$ and $v$ be the endpoints
of $h$ and $h'$, that are not incident to $f$, respectively.
Consider the cubic pseudo-graph $G'_{0}$ obtained from $G_{0}$ as
follows ((a)
of figure \ref{loopcut}):%
\begin{eqnarray*}
G'_{0}=(G_{0}\backslash \{u_{0},v_{0}\})\cup \{g\},\text{ where } \
g=(u,v).
\end{eqnarray*}%

\begin{figure}[h]
\begin{center}
\includegraphics{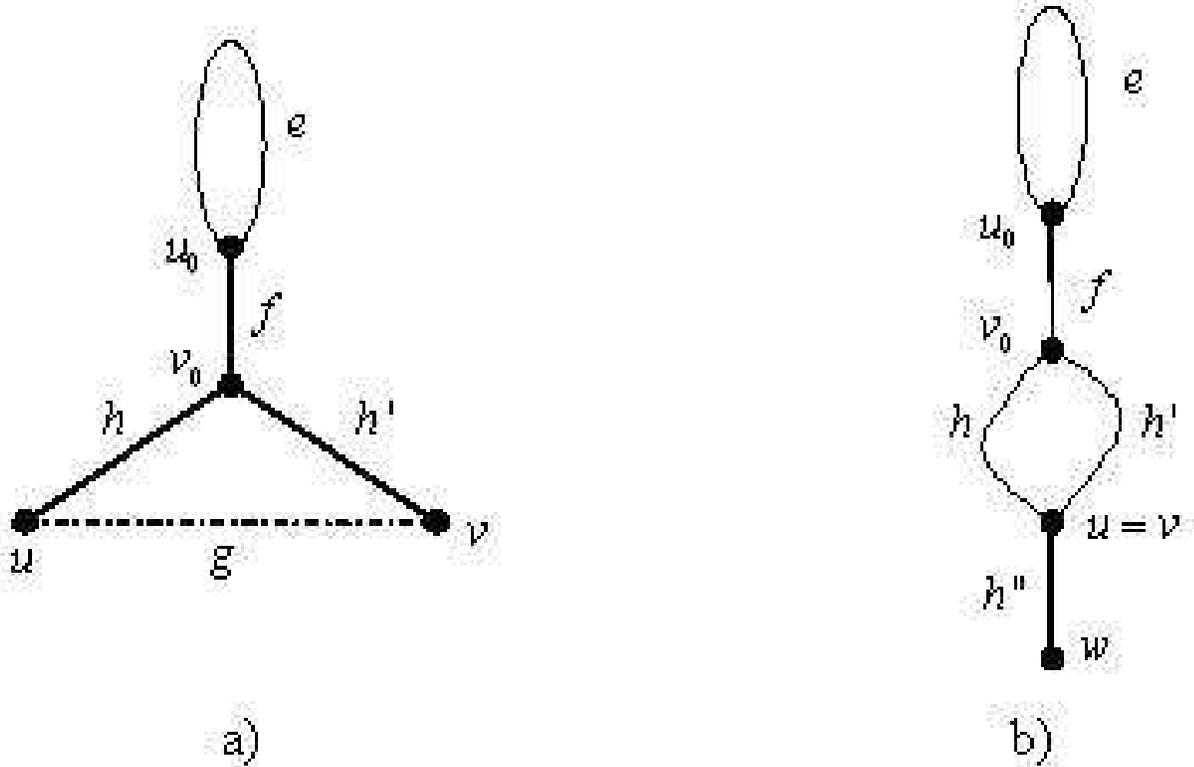}\\
\caption{Cutting a loop $e$}\label{loopcut}
\end{center}
\end{figure}

Note that $u$ and $v$ may coincide. In this case $g$ is a loop of $%
G'_{0}$. We will say that $G'_{0}$ is obtained from $G_{0}$ by
cutting the loop $e$.

People dealing especially with bridgeless cubic graphs would have
already recognized Fleischner's splitting off operation. Completely
realizing this, we would like to keep the name "cutting the loops",
in order to keep the basic idea, that has helped us to come to its
definition!

\begin{remark}
\label{SuccessiveCut}If $G_{0}$ is a connected cubic pseudo-graph,
then the successive cut of loops of $G_{0}$ in any order of loops
leads either to a connected graph (that is, connected pseudo-graph
without loops), or to the cubic pseudo-graph shown on the figure
\ref{TrivialCase}. Sometimes, we will prefer to restate this
property in terms of applicability of the operation of cutting the
loop. More specifically, if $G_{0}$ is a connected cubic
pseudo-graph, for which the operation of cutting the loop is not
applicable, then either $G_{0}$ does not contain a loop or it is the
mentioned trivial graph.
\end{remark}

\begin{figure}[h]
\begin{center}
\includegraphics[height=5pc,width=15pc]{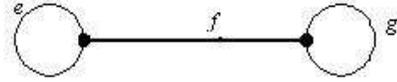}\\
\caption{The trivial case}\label{TrivialCase}
\end{center}
\end{figure}

Before we move on, we would like to state some properties of the
operation of cutting the loops.

\begin{proposition}
\label{ConnectedInvariance}If $G_{0}$ is connected, then $G'_{0}$ is
connected, too.
\end{proposition}

\begin{proposition}
\label{CycleInvariance}If a connected cubic pseudo-graph $G_{0}$
contains a cycle, and a cubic pseudo-graph $G'_{0}$ obtained from
$G_{0}$ by cutting a loop $e$ of $G_{0}$ does not, then $e$ is
adjacent to an edge $f$, which, in its turn, is adjacent to two
edges $h$ and $h'$, that form the only cycle of $G_{0}$ with length
two ((b) of figure \ref{loopcut}).
\end{proposition}

The following will be used frequently:

\begin{proposition}
\label{FractionInequality}Let be $a,b,c,d$ be positive numbers with $\frac{a%
}{b}\geq \alpha $, $\frac{c}{d}\geq \alpha $. Then:%
\begin{equation}
\frac{a+c}{b+d}\geq \alpha.
\end{equation}
\end{proposition}

\begin{proposition}
\label{LinearInequality}Suppose that $x_{1}\leq y_{1},$ $x_{2}\geq
y_{2},...,x_{n}\geq y_{n},$ $x_{1}+...+x_{n}=y_{1}+...+y_{n}$ and %
${\min }_{{1\leq i\leq n}}\{\alpha _{i}\}=\alpha _{1}>0$. Then:%
\begin{equation}
\alpha _{1}x_{1}+...+\alpha _{n}x_{n}\geq \alpha
_{1}y_{1}+...+\alpha _{n}y_{n}.
\end{equation}
\end{proposition}

\begin{proof}
Note that
\begin{eqnarray*}
\alpha _{2}(x_{2}-y_{2}) &\geq &\alpha _{1}(x_{2}-y_{2}), \\
&&. \\
&&. \\
&&. \\
\alpha _{n}(x_{n}-y_{n}) &\geq &\alpha _{1}(x_{n}-y_{n}),
\end{eqnarray*}%
thus
\begin{equation}
\alpha _{2}(x_{2}-y_{2})+...+\alpha _{n}(x_{n}-y_{n})\geq \alpha
_{1}(x_{2}-y_{2})+...+\alpha _{1}(x_{n}-y_{n})=\alpha
_{1}(y_{1}-x_{1})
\end{equation}%
or%
\begin{equation}
\alpha _{1}x_{1}+...+\alpha _{n}x_{n}\geq \alpha
_{1}y_{1}+...+\alpha _{n}y_{n}.
\end{equation}
\end{proof}

\begin{theorem}\label{Gallai}(Gallai \cite{Lov})Let $G$ be a connected graph
with $\nu(G-u)=\nu$ for any $u\in V$. Then $G$ is factor-critical,
and particularly:
\begin{equation*}
n=2\nu+1.
\end{equation*}
\end{theorem}
Terms and concepts that we do not define can be found in
\cite{Harary,Lov,West}.

\section{Maximum matchings and unsaturated vertices}

In this section we prove a lemma, which states that, under some
conditions, one can always pick up a maximum matching of a graph,
such that the unsaturated vertices with respect to this matching
"are not placed very close".

Before we present our result, we would like to deduce a lower bound
for $\nu $ in the class of regular graphs using the theorem
\ref{Gallai} of Gallai.

Observe that Shannon's theorem implies that $\chi ^{\prime }\leq 4$ for every cubic graph $G$, thus $%
m \leq 4\nu =4\nu _{1}$. Now, it turns out, that there are no cubic
graphs $G$, for which $m =4\nu _{1}$, thus $\nu _{1}>\frac{m}{4}$.
Next we
prove a generalization of this statement, that originally appeared in \cite%
{Monthly} as a problem:

\begin{lemma}
\label{OddRegulars}

\begin{description}
\item[(a)] No $(2k+1)$-regular graph $G$ contains $2k+2$ pairwise
edge-disjoint maximum matchings;

\item[(b)] If $G$ is a connected simple $r$-regular graph with $r+1$
pairwise edge-disjoint maximum matchings, then $r$ is even and $G$
is the complete graph.
\end{description}
\end{lemma}

\begin{proof}
(a) Assume $G$ to contain $2k+2$ pairwise edge-disjoint maximum matchings $%
F_{1},...,F_{2k+2}$. Note that we may assume $G$ to be connected.
Clearly, for every $v\in V$ there is $F_{v}\in \left\{
F_{1},...,F_{2k+2}\right\} $ such that $F_{v}$ does not saturate the
vertex $v$. By a theorem \ref{Gallai} of Gallai, it follows that $n
=2\nu+1$, that is, $n$ is odd, which is impossible.

(b) Assume $G$ to contain $r+1$ pairwise edge-disjoint maximum matchings $%
F_{1},...,F_{r+1}$. (a) implies that $n =2\nu+1$ and $r$ is even. Since, by Vizing's theorem $\chi ^{\prime }\leq r+1$, we have%
\begin{equation*}
(r+1)\nu =\left\vert F_{1}\right\vert +...+\left\vert
F_{r+1}\right\vert \leq m \leq \chi ^{\prime }\nu \leq (r+1)\nu,
\end{equation*}
thus
\begin{equation*}
(r+1)\frac{n-1}{2}=(r+1)\nu =m =r\frac{n}{2},
\end{equation*}
or
\begin{equation*}
r=n -1,
\end{equation*}
hence $G$ is the complete graph.
\end{proof}

\begin{remark}
As the example of the "fat triangle" shows, the complete graph with
odd number of vertices is not the only graph, that prevents us to
generalize (a) to even regular graphs.
\end{remark}

Next we prove the main result of the section, which is interesting
not only on its own, but also will help us to derive better bounds
in the theorem \ref{MainTheoremCubics}.

\begin{lemma}
\label{Max Matching 2-3} Every graph $G$, with $2\leq \delta\leq
\Delta \leq 3$, contains a maximum matching, such that the
unsaturated vertices (with respect to this maximum matching) do not
share a neighbour.
\end{lemma}

\begin{proof}
Let $F$ be a maximum matching of $G$, for which there are minimum
number of pairs of unsaturated vertices, which have a common
neighbour. The lemma will be proved, if we show that this number is
zero.

Suppose that there are vertices $u$ and $w$ of $G$ which are not
saturated (by $F$) and have a common neighbour $q$. Clearly, $q$ is
saturated by an edge $e_q\in F$. Consider the edge $e=(u,q)$. Note
that it lies in a maximum matching of $G$ (an example of such a
maximum matching is $(F\backslash \{e_q\})\cup \{e\}$). Moreover,
for every maximum matching $F_{e}$
of $G$ with $e\in F_{e}$, the alternating component $P_{e}$ of $%
F\bigtriangleup F_{e}$ which contains the edge $e$, is a path of
even length. Now, choose a maximum matching $F^{\prime }$ of $G$
containing the edge $e$ for which the length of $P_{e}$ is maximum.

Let $v$ be the other ($\neq u$) end-vertex of the path $P_{e}$. Note
that since $P_{e}$ is even, there is a vertex $p$ of $P_{e}$ such that $(p,v)\in F$%
.

\begin{claim}
\label{Neighbours V}The neighbours of $v$ lie on $P_{e}$ and are
different from $u$ and $q$.
\end{claim}

\begin{proof}
First of all let us show that the neighbours of $v$ lie on $P_{e}$.
On the opposite assumption, consider a vertex $v^{\prime }$ which is
adjacent to $v$ and which does not lie on $P_{e}$. Clearly
$(v,v^{\prime })\notin F\cup F^{\prime }$. As $F^{\prime }$ is a
maximum matching, there is an edge $f\in F^{\prime }$ incident to
$v^{\prime }$. Define:
\begin{equation*}
F^{\prime \prime }=(F^{\prime }\backslash \{f\})\cup \{(v,v^{\prime
})\}.
\end{equation*}%
Note that $F^{\prime \prime }$ is a maximum matching of $G$ with
$e\in
F^{\prime \prime }$ for which the length of the alternating component of $%
F\bigtriangleup F^{\prime \prime }$, which contains the edge $e$,
exceeds the length of $P_{e}$ contradicting the choice of $F^{\prime
}$. Thus the
neighbours of $v$ lie on $P_{e}$. Let us show that they are different from $%
u $ and $q$. If there is an edge $e_{1}$ connecting the vertices $u$ and $v$%
, then define:
\begin{equation*}
F^{\prime \prime \prime }=(F\backslash E(P_{e}))\cup ([F^{\prime
}\cap E(P_{e})]\backslash \{e\})\cup \{e_{1},(q,w)\}.
\end{equation*}%
Clearly, $F^{\prime \prime \prime }$ is a matching of $G$ for which $%
\left\vert F^{\prime \prime \prime }\right\vert >\left\vert
F\right\vert $,
which is impossible. Thus, there are no edges connecting $u$ and $v$. As $%
q$ is adjacent to $u$ and $w$, $v$ can be adjacent to $q$ if and
only if $p=q$, that is, if the length of $P_{e}$ is two. But this is
impossible, too, since $d_{G}(v)\geq 2$, hence there should be an
edge connecting $u$ and $v$. The proof of claim \ref{Neighbours V}
is completed.
\end{proof}

\begin{corollary}
The length of $P_{e}$ is at least four.
\end{corollary}

To complete the proof of the lemma we need to consider two cases:

Case 1: $(p,w)\notin E$.

Consider a maximum matching $F_{0}$ of $G$ which is obtained from
$F$ by
shifting the edges of $F$ on $P_{e}$, that is,%
\begin{equation*}
F_{0}=(F\backslash E(P_{e}))\cup (F^{\prime }\cap E(P_{e})).
\end{equation*}

Note that $F_{0}$ saturates all vertices of $P_{e}$ except $v$.
Consider a vertex $v_{0}$ which is a neighbour of $v$. Due to claim
\ref{Neighbours V}, $v_{0}$ is a vertex of $P_{e}$, which is
different from $u$ and $q$. Note that the neighbours of $v_{0}$ are
the vertex $v$ and one or two other vertices of $P_{e}$ which are
saturated by $F_{0}$. Thus there is no unsaturated vertex of $G$,
which has a common neighbour with $v$. This implies that the number
of pairs of vertices of $G$ which are not saturated by $F_{0}$ and
have a common neighbour is less than the corresponding number for
$F$, which contradicts the choice of $F$.

Case 2: $(p,w)\in E$.

Consider a maximum matching $F_{1}$ of $G$, defined as:%
\begin{equation*}
F_{1}=(F\backslash \left\{ (p,v)\right\} )\cup \left\{ (p,w)\right\}
.
\end{equation*}

Note that $F_{1}$ saturates $w$ and does not saturate $v$. Consider
a vertex $v_{1}$ which is a neighbour of $v$. Due to claim
\ref{Neighbours V}, $v_{1}$ is a vertex of $P_{e}$, which is
different from $u$ and $q$. Note that the neighbours of $v_{1}$ are
the vertex $v$ and two other vertices of $P_{e}$
which are saturated by $F_{1}$. Thus there is no unsaturated vertex of $G$%
, which has a common neighbour with $v$. This implies that the
number of pairs of vertices of $G$ which are not saturated by
$F_{1}$ and have a common neighbour is less than the corresponding
number for $F$, which contradicts the choice of $F$. The proof of
lemma \ref{Max Matching 2-3} is completed.
\end{proof}

It would be interesting to generalize the statement of lemma
\ref{Max Matching 2-3} to almost regular graphs. In other words, we
would like to suggest the following

\begin{conjecture}
Let $G$ be graph with $\Delta-\delta\leq 1$. Then $G$ contains a
maximum matching such that the unsaturated vertices (with respect to
this maximum matching) do not share a neighbour.
\end{conjecture}

We would like to note that we do not even know, whether the
conjecture holds for $r$-regular graphs with $r\geq 4$.

\section{The system of cycles and paths}

In this section we prove two lemmas. For graphs that belong to a
very peculiar family, the first of them allows us to find a system
of cycles and paths that satisfy some explicitly stated properties.
The second lemma helps in finding a system with the same properties
in graphs that are subdivisions of the graphs from the mentioned
peculiar class. Moreover, due to the second lemma, it turns out that
if there is a system of the original graph that includes a maximum
matching, then there is a system of the subdivided graph preserving
this property!

\begin{lemma}
\label{Bipartite 2->=3}Let $G$ be a graph with $\delta \geq 2$.
Suppose that every edge of $G$ connects a vertex of degree two to
one with degree at least three. Then
\end{lemma}

\begin{enumerate}
\item[(1)] \textit{there exists a vertex-disjoint system of even paths }$%
P_{1},...,P_{r}$\textit{\ and cycles }$C_{1},...,C_{l}$\textit{\ of }$G$%
\textit{\ such that }

\begin{enumerate}
\item[(1.1)] $r=\frac{1}{2}\sum_{v,d(v)\geq 3}
(d(v)-2);$

\item[(1.2)] \textit{all vertices of }$G$\textit{\ lie on these paths or
cycles;}

\item[(1.3)] \textit{the end-vertices of the paths }$P_{1},...,P_{r}$\textit{%
\ are of degree two and these end-vertices are adjacent to vertices
of degree at least three;}

\end{enumerate}

\item[(2)] \textit{for every maximum matching }$F$\textit{\ of }$G$\textit{,
every pair of edge-disjoint matchings }$(H,H^{\prime })$\textit{\ with }$%
\left\vert H\right\vert +\left\vert H^{\prime }\right\vert =\nu _{2},$%
\textit{\ every vertex }$v\in V$\textit{\ with }$d(v)\geq 3,$%
\textit{\ }\textit{\ is incident to one edge from }$F$\textit{, one from }%
$H$\textit{\ and one from }$H^{\prime }$\textit{.}

\item[(3)] $G$ \textit{contains two edge-disjoint maximum matchings};

\item[(4)] \label{NuIN(2,k)graphs}\textit{If} $\delta =2$, $\Delta =k\geq 3$, $%
d(v)\in \{2,k\}$ \textit{for every vertex} $v\in V,$ \textit{then}

\begin{equation}
\nu _{1}=\frac{2}{k+2}n ,\nu _{2}=\frac{4}{k+2%
}n.  \label{eq3}
\end{equation}
\end{enumerate}

\begin{proof}
(1) Clearly, $G$ is a bipartite graph, since the sets%
\begin{eqnarray*}
V_{2} &=&\left\{ v\in V:d(v)=2\right\} , \\
V_{\geq 3} &=&\left\{ v\in V:d(v)\geq 3\right\}
\end{eqnarray*}

form a bipartition of $G$. We intend to construct a system of
pairwise
vertex-disjoint cycles and even paths of $G$ such that the all vertices of $%
V_{\geq 3}$ lie on them. Of course, the cycles will be of even length since $%
G$ is bipartite.

Choose a system of cycles $C_{1},...,C_{l}$ of $G$ such that
$V(C_{i})\cap V(C_{j})=\emptyset$, $1\leq i<j\leq l$ and the graph
$G_{0}=G\backslash (V(C_{1})\cup ...V(C_{l}))$ does not contain a
cycle. Clearly, $G_{0}$ is a forest, that is, a graph every
component of which is a tree. Moreover, for every $v_{0}\in V_{0}$

\begin{description}
\item[(a)] if $d_{G_{0}}(v_{0})\geq 3$ then $d_{G_{0}}(v_{0})=d_{G}(v_{0})$;

\item[(b)] if $d_{G_{0}}(v_{0})\in \left\{ 0,1,2\right\} $ then $%
d_{G}(v_{0})=2$.
\end{description}

If $G_{0}$ contains no edge then add the remaining isolated vertices
(paths of length zero) to the system to obtain the mentioned system
of cycles and even paths of $G$. Otherwise, consider a non-trivial
component of $G_{0}$. Let $P_{1}$ be a path of this component
connecting two vertices which have degree one in $G_{0}$. Since $G$
is bipartite, (b) implies that $P_{1}$ is of even length. Consider a
graph $G_{1}$ obtained from $G_{0}$ by removing
the path $P_{1}$, that is,%
\begin{equation*}
G_{1}=G_{0}\backslash V(P_{1}).
\end{equation*}

Note that $G_{1}$ is a forest. Moreover, it satisfies the properties
(a) and (b) like $G_{0}$ does, that is, for every $v_{1}\in V_{1}$

\begin{description}
\item[(a$^{\prime }$)] if $d_{G_{1}}(v_{1})\geq 3$ then $%
d_{G_{1}}(v_{1})=d_{G}(v_{1})$;

\item[(b$^{\prime }$)] if $d_{G_{1}}(v_{1})\in \left\{ 0,1,2\right\} $ then $%
d_{G}(v_{1})=2$.
\end{description}

Clearly, by the repeated application of this procedure we will get a
system of even paths $P_{1},...,P_{r_{0}}$ of $G$ such that the
graph $G\backslash (V(C_{1})\cup ...V(C_{l})\cup V(P_{1})\cup
...V(P_{r_{0}}))=G_{0}\backslash (V(P_{1})\cup ...V(P_{r_{0}}))$
contains no edge. Now, add the remaining isolated vertices (paths of
length zero) to $P_{1},...,P_{r_{0}}$ to obtained a system of even
paths $P_{1},...,P_{r}$.

Note that by the construction $C_{1},...,C_{l}$ and
$P_{1},...,P_{r}$ are vertex-disjoint. Moreover, the paths
$P_{1},...,P_{r}$ are of even length. As $G$ is bipartite, the
cycles $C_{1},...,C_{l}$ are of even length, too.

Again, by the construction of $C_{1},...,C_{l}$ and
$P_{1},...,P_{r}$ we have (1.2) and that the end-vertices of
$P_{1},...,P_{r}$ are of degree two. As every edge of $G$ connects a
vertex of degree two to one with degree at least three, the system
$C_{1},...,C_{l},P_{1},...,P_{r}$ satisfies (1.3).

Let us show that (1.1) holds, too. Since the number of vertices of
degree two and at least three is the same on the cycles
$C_{1},...,C_{l}$, and the difference of these two numbers is one on
each path from $P_{1},...,P_{r}$,
then taking into account (1.2) and (1.3) we get:%
\begin{equation*}
r=\left\vert V_{2}\right\vert -\left\vert V_{\geq 3}\right\vert =\frac{%
2\left\vert V_{2}\right\vert -2\left\vert V_{\geq 3}\right\vert }{2}=\frac{%
\sum_{v,d(v)\geq 3}d(v)-2\left\vert V_{\geq 3}\right\vert
}{2}=\frac{1}{2}\sum_{v,d(v)\geq 3}(d(v)-2).
\end{equation*}

(2) Define a pair of edge-disjoint matchings $(H_{0},H_{0}^{\prime
})$ in
the following way: alternatively add the edges of $C_{1},...,C_{l}$ and $%
P_{1},...,P_{r}$ to $H_{0}$ and $H_{0}^{\prime }$. Note that every vertex $%
v\in V_{\geq 3}$ is incident to one edge from $H_{0}$, one from $%
H_{0}^{\prime }$, and%
\begin{equation}
2\nu _{1}\geq \nu _{2}\geq \left\vert H_{0}\right\vert +\left\vert
H_{0}^{\prime }\right\vert =2\left\vert V_{\geq 3}\right\vert .
\label{eq1}
\end{equation}

On the other hand, for every pair of edge-disjoint matchings
$(h,h^{\prime }) $, every vertex $v\in V_{\geq 3}$ is incident to at
most one edge
from $h$ and at most two edges from $h\cup h^{\prime }$, therefore%
\begin{eqnarray*}
\nu _{1} &=&\underset{h}{\max }\left\vert h\right\vert \leq
\left\vert
V_{\geq 3}\right\vert , \\
\nu _{2} &=&\underset{h\cap h^{\prime }=\emptyset}{\max }(\left\vert
h\right\vert +\left\vert h^{\prime }\right\vert )\leq 2\left\vert
V_{\geq 3}\right\vert ,
\end{eqnarray*}

thus (see (\ref{eq1}))%
\begin{equation}
\nu _{1}=\left\vert V_{\geq 3}\right\vert ,\nu _{2}=2\left\vert
V_{\geq 3}\right\vert ,  \label{eq2}
\end{equation}

and for every maximum matching $F$ of $G$, every pair of
edge-disjoint matchings $(H,H^{\prime })$ with $\left\vert
H\right\vert +\left\vert H^{\prime }\right\vert =\nu _{2}$, every
vertex $v\in V_{\geq 3}$ is incident to one edge from $F$, one from
$H$ and one from $H^{\prime }$.

(3) directly follows from (2). (4) follows from (2) and the
bipartiteness of $G$.

The proof of the lemma \ref{Bipartite 2->=3} is completed.
\end{proof}

\begin{lemma}
\label{SystemInSubdivision}Let $G$ be a connected graph satisfying
the conditions:
\end{lemma}

\begin{description}
\item[(a)] $\delta \geq 2$;

\item[(b)] \textit{no edge of }$G$\textit{\ connects two vertices having degree at
least three.}

\item[ ] \textit{Let }$G^{\prime }$\textit{\ be a graph obtained from }$G$%
\textit{\ by a }$1$\textit{-subdivision of an edge. If }$G$\textit{\
contains a system of paths }$P_{1},...,P_{r}$\textit{\ and even cycles }$%
C_{1},...,C_{l}$\textit{\ such that }

\begin{enumerate}
\item[(1)] \textit{the degrees of vertices of a cycle from }$C_{1},...,C_{l}$%
\textit{\ are two and at least three alternatively,}

\item[(2)] \textit{all vertices of }$G$\textit{\ lie on these paths or
cycles;}

\item[(3)] \textit{the end-vertices of the paths }$P_{1},...,P_{r}$\textit{\
are of degree two, and the vertices that are adjacent to these
end-vertices and do not lie on }$P_{1},...,P_{r}$ \textit{\ are of
degree at least three;}

\item[(4)] \textit{every edge that does not lie on }$C_{1},...,C_{l}$\textit{%
\ and }$P_{1},...,P_{r}$\textit{\ is incident to one vertex of
degree two and one of degree at least three;}

\item[(5)] \textit{there is a maximum matching }$F$\textit{\ of }$G$\textit{%
\ such that every edge }$e\in F$\textit{\ lies on }$C_{1},...,C_{l}$\textit{%
\ and }$P_{1},...,P_{r}$,
\end{enumerate}

\item \textit{then there is} \textit{a system of paths }$P_{1}^{\prime
},...,P_{r^{\prime }}^{\prime }$\textit{\ and even cycles
}$C_{1}^{\prime },...,C_{l^{\prime }}^{\prime }$\textit{\ of the
graph }$G^{\prime }$ \textit{with} $r^{\prime }=r$
\textit{satisfying (1)-(5)}.
\end{description}

\begin{proof}
Let $P_{1},...,P_{r}$\textit{\ }and $C_{1},...,C_{l}$ be a system of
paths
and even cycles satisfying (1)-(5) and let $e$ be the edge of $G$ whose $1$%
-subdivision led to the graph $G^{\prime }$. First of all we will
construct a system of paths and even cycles of $G^{\prime }$
satisfying the conditions (1)-(4).

We need to consider three cases:

Case 1: $e$ lies on a path $P\in \{P_{1},...,P_{r}\}$.

Let $P^{\prime }$ be the path of $G^{\prime }$ corresponding to $P$
(that is, the path obtained from $P$ by the $1$-subdivision of the
edge $e$). Consider a system of paths and even cycles of $G^{\prime
}$ defined as:
\begin{eqnarray*}
C_{i}^{\prime } &=&C_{i},i=1,...,l, \\
\{P_{1}^{\prime },...,P_{r^{\prime }}^{\prime }\}
&=&(\{P_{1},...,P_{r}\}\backslash \{P\})\cup \{P^{\prime }\}
\end{eqnarray*}
Clearly, $r^{\prime }=r$. It can be easily verified that the system $%
P_{1}^{\prime },...,P_{r^{\prime }}^{\prime }$\textit{\ }and\textit{\ }$%
C_{1}^{\prime },...,C_{l}^{\prime }$ satisfies (1)-(4).

Case 2: $e$ does not lie on either of $P_{1},...,P_{r}$\textit{\ }and $%
C_{1},...,C_{l}$ .

Let $w_{e}$ be the new vertex of $G^{\prime }$ and let $e^{\prime
},e^{\prime \prime }$ be the new edges of $G^{\prime }$, that is:
\begin{eqnarray*}
V' &=&V\cup \{w_{e}\}, \\
E' &=&(E\backslash \{e\})\cup \{e^{\prime },e^{\prime \prime }\}.
\end{eqnarray*}
(4) implies that $e$ is incident to a vertex $u$ of degree two and a vertex $%
v$ of degree at least three, and suppose that $e^{\prime
}=(v,w_{e}),e^{\prime \prime }=(w_{e},u)$.\newline

Since $d(u)=2$ and $e$ does not lie on $P_{1},...,P_{r}$\textit{\ }and $%
C_{1},...,C_{l}$ , (2) implies that there is a path $P_{u}\in
\{P_{1},...,P_{r}\}$ such that $u$ is an end-vertex of $P_{u}$.
Consider the path $P_{u}^{\prime }$ defined as:
\begin{equation*}
P_{u}^{\prime }=w_{e},e^{\prime \prime },P_{u}
\end{equation*}%
and a system of paths and even cycles of $G^{\prime }$ defined as:
\begin{eqnarray*}
C_{i}^{\prime } &=&C_{i},i=1,...,l, \\
\{P_{1}^{\prime },...,P_{r^{\prime }}^{\prime }\}
&=&(\{P_{1},...,P_{r}\}\backslash \{P_{u}\})\cup \{P_{u}^{\prime
}\}.
\end{eqnarray*}%
Clearly, $r^{\prime }=r$. Note that the new system satisfies (1) and
(2). Let us show that it satisfies (3) and (4), too. Since
$d_{G^{\prime }}(w_{e})=2$, $w_{e}$ is adjacent to the vertex $v$ of
degree at least three
and $w_{e}$ is an end-vertex of $P_{u}^{\prime }$, we imply that the system $%
P_{1}^{\prime },...,P_{r^{\prime }}^{\prime }$\textit{\ }and\textit{\ }$%
C_{1}^{\prime },...,C_{l}^{\prime }$ satisfies (3).

Note that we need to verify (4) only for the edge $e^{\prime }$. As $%
d_{G^{\prime }}(w_{e})=2$, $d_{G^{\prime }}(v)\geq 3$, we imply that
the
system $P_{1}^{\prime },...,P_{r^{\prime }}^{\prime }$\textit{\ }and\textit{%
\ }$C_{1}^{\prime },...,C_{l}^{\prime }$ satisfies (4), too.

Case 3: $e$ lies on a cycle $C\in \{C_{1},...,C_{l}\}$.

Let $w_{e}$ be the new vertex of $G^{\prime }$ and let $e^{\prime
},e^{\prime \prime }$ be the new edges of $G^{\prime }$, that is:
\begin{eqnarray*}
V' &=&V\cup \{w_{e}\}, \\
E' &=&(E\backslash \{e\})\cup \{e^{\prime },e^{\prime \prime }\}.
\end{eqnarray*}

(1) implies that the edge $e$ is incident to a vertex $u$ of degree
two and to a vertex $v$ of degree at least three, and suppose that
$e^{\prime }=(v,w_{e})$, $e^{\prime \prime }=(w_{e},u)$. Since
$d_{G}(v)\geq 3$, (b)
implies that there is a vertex $z\in V$ such that $(v,z)\in E$ and $%
z\notin V(C)$. Note that since $d_{G}(z)=2$ and the edge $(v,z)$
does not lie on either of $P_{1},...,P_{r}$\textit{\ }and
$C_{1},...,C_{l}$ (2) implies that there is a path $P_{z}\in
\{P_{1},...,P_{r}\}$ such that $z$ is an end-vertex of $P_{z}$.

Let $P$ be the path $C-e$ of $G$ starting from the vertex $v$.
Consider a path $P^{\prime }$ of $G^{\prime }$ defined as:
\begin{equation*}
P^{\prime }=P_{z},(z,v),P,e^{\prime \prime },w_{e}
\end{equation*}%
and a system of paths and even cycles of $G^{\prime }$ defined as:
\begin{eqnarray*}
\{C_{1}^{\prime },...,C_{l^{\prime }}^{\prime }\}
&=&(\{C_{1},...,C_{l}\}\backslash \{C\}) \\
\{P_{1}^{\prime },...,P_{r^{\prime }}^{\prime }\}
&=&(\{P_{1},...,P_{r}\}\backslash \{P_{z}\})\cup \{P^{\prime }\}.
\end{eqnarray*}%
Clearly, $r^{\prime }=r$. Note that the new system satisfies (1) and
(2). Let us show that it satisfies (3) and (4), too. Since
$d_{G^{\prime }}(w_{e})=2$, $w_{e}$ is adjacent to the vertex $v$ of
degree at least three, we imply that the system $P_{1}^{\prime
},...,P_{r^{\prime }}^{\prime }$\textit{\ }and\textit{\
}$C_{1}^{\prime },...,C_{l}^{\prime }$ satisfies (3).

Note that we need to verify (4) only for the edge $e^{\prime }$. As $%
d_{G^{\prime }}(w_{e})=2$, $d_{G^{\prime }}(v)\geq 3$ we imply that
the
system $P_{1}^{\prime },...,P_{r^{\prime }}^{\prime }$\textit{\ }and\textit{%
\ }$C_{1}^{\prime },...,C_{l}^{\prime }$ satisfies (4), too.

The consideration of these three cases implies that there is a system $%
P_{1}^{\prime },...,P_{r^{\prime }}^{\prime }$\textit{\ }and\textit{\ }$%
C_{1}^{\prime },...,C_{l^{\prime }}^{\prime }$ of paths and even cycles of $%
G^{\prime }$ with $r^{\prime }=r$ satisfying the conditions (1)-(4).
Let us show that among such systems there is at least one satisfying
(5), too.

Consider all pairs $(\mathfrak{F'}_{0},M'_{0})$ in the
graph $G'$ where $\mathfrak{F'}_{0}$ is a system $%
P_{1}^{\prime },...,P_{r^{\prime }}^{\prime }$\textit{\ }and\textit{\ }$%
C_{1}^{\prime },...,C_{l^{\prime }}^{\prime }$ of paths and even cycles of $%
G^{\prime }$ with $r^{\prime }=r$ satisfying the conditions (1)-(4) and $%
M_{0}^{\prime }$ is a maximum matching of $G^{\prime }$. Among these
choose a pair $(\mathfrak{F}^{\prime },M^{\prime })$ for which the
number of edges of $M^{\prime }$ which lie on cycles and paths of
$\mathfrak{F}^{\prime }$ is maximum.We claim that all edges of
$M^{\prime }$ lie on cycles and paths of $\mathfrak{F}^{\prime }$.

\begin{claim}
\label{CycleCase}If $C$ is a cycle from $\mathfrak{F}^{\prime }$
with length $2n$ then there are exactly $n$ edges of $M^{\prime }$
lying on $C$.
\end{claim}

\begin{proof}
Let $k$ be the number of vertices of $C$ which are saturated by an
edge from $M^{\prime }\backslash E(C)$. (1) implies that if we
remove these $k$ vertices from $C$ we will get $k$ paths with an odd
number of vertices. Thus
each of these $k$ paths contains a vertex that is not saturated by $%
M^{\prime }$. Thus the total number of edges from $M^{\prime }\cap
E(C)$ is at most
\begin{equation*}
\left\vert M^{\prime }\cap E(C)\right\vert \leq
\frac{2n-2k}{2}=n-k\text{.}
\end{equation*}%
Consider a maximum matching $M^{\prime }$ of $G^{\prime }$ defined
as:
\begin{equation*}
M^{\prime \prime }=(M^{\prime }\backslash M_{C}^{\prime })\cup
M_{C}^{\prime \prime }\text{,}
\end{equation*}%
where $M_{C}^{\prime }$ is the set of edges of $M^{\prime }$ that
are
incident to a vertex of $C$, and $M_{C}^{\prime \prime }$ is a 1-factor of $%
C $. Note that if $k\geq 1$ then%
\begin{equation*}
\left\vert M^{\prime \prime }\cap E(C)\right\vert >\left\vert
M^{\prime }\cap E(C)\right\vert
\end{equation*}%
and therefore for the pair $(\mathfrak{F}^{\prime },M^{\prime \prime
})$ we would have that $M^{\prime \prime }$ contains more edges
lying on cycles and paths of $\mathfrak{F}^{\prime }$ then
$M^{\prime }$ does, contradicting the choice of the pair
$(\mathfrak{F}^{\prime },M^{\prime })$, thus $k=0$, and on the cycle
$C$ from $\mathfrak{F}^{\prime }$ with length $2n$ there are exactly
$n$ edges of $M^{\prime }$. The proof of claim \ref{CycleCase} is
completed.
\end{proof}

Now, we are ready to prove that all edges of $M^{\prime }$ lie on
cycles and paths of $\mathfrak{F}^{\prime }$. Suppose, on the
contrary, that there is an edge $e^{\prime }\in M^{\prime }$ that
does not lie on cycles and paths of $\mathfrak{F}^{\prime }$. (4)
implies that $e^{\prime }$ is incident to a vertex $u$ of degree at
least three and to a vertex $v$ of degree two. (2)
implies that there is a path $P_{v}$ of $\mathfrak{F}^{\prime }$ such that $%
v $ is an end-vertex of $P_{v}$. (2) and claim \ref{CycleCase} imply
that
there is a path $P_{u}$ of $\mathfrak{F}^{\prime }$ such that $u$ lies on $%
P_{u}$. Let $w$ and $z$ be the end-vertices of $P_{u}$, and let $P_{wu}$ and $%
P_{zu}$ be the subpaths of the path $P_{u}$ connecting $w$ and $z$
to $u$, respectively. Consider a system $\mathfrak{F}^{\prime \prime
}$ of paths and even cycles of $G^{\prime }$ defined as follows:
\begin{equation*}
\mathfrak{F}^{\prime \prime }=(\mathfrak{F}^{\prime }\backslash
\{P_{u},P_{v}\})\cup \{P_{zu}-u,P^{\prime }\}
\end{equation*}%
where the path $P^{\prime }$ is defined as:%
\begin{equation*}
P^{\prime }=P_{wu},(u,v),v,P_{v}\text{.}
\end{equation*}%
\ Note that $\mathfrak{F}^{\prime \prime }$ contains exactly
$r^{\prime }=r$ paths. It can be easily verified that the new system
$\mathfrak{F}^{\prime \prime }$ of paths and even cycles of
$G^{\prime }$ satisfies (1)-(4).

Now if we consider the pair $(\mathfrak{F}^{\prime \prime
},M^{\prime })$ we would have that the paths and even cycles of
$\mathfrak{F}^{\prime \prime }$
include more edges of $M^{\prime }$ then the paths and even cycles of $%
\mathfrak{F}^{\prime }$ do, contradicting the choice of the pair $(\mathfrak{%
F}^{\prime },M^{\prime })$. Thus, all edges of $M^{\prime }$ lie on
cycles
and paths of $\mathfrak{F}^{\prime }$. The proof of the lemma \ref%
{SystemInSubdivision} is completed.
\end{proof}

\section{The subdivision and the main parameters}

The aim of this section is to prove a lemma, which claims that,
under some conditions, the subdivision of an edge increases the size
of the maximum 2-edge-colorable subgraph of a graph by one. This is
important for us, since it enables us to control our parameters,
while considering many graphs that are subdivisions of the others.

\begin{lemma}
\label{Edge Subdivision}Let $G$ be a connected graph satisfying the
conditions:
\end{lemma}

\begin{description}
\item[(a)] $\delta \geq 2$;

\item[(b)] $G$\textit{\ is not an even cycle;}

\item[(c)] \textit{no edge of }$G$\textit{\ connects vertices with degree at
least three.}

\item[ ] \textit{Let }$G^{\prime }$\textit{\ be a graph obtained from }$G$%
\textit{\ by a }$1$\textit{-subdivision of an edge. Then}
\end{description}

\begin{enumerate}
\item[(1)] $\nu' _{2}\geq 1+\nu _{2}$;

\item[(2)] $\nu' _{2}=\left\{
\begin{array}{ll}
2+\nu _{2}, & \text{if }G\text{ is an odd cycle,} \\
1+\nu _{2}, & \text{otherwise.}%
\end{array}%
\right. $
\end{enumerate}

\begin{proof}
(1) Let $(H,H^{\prime })$ be a pair of edge-disjoint matchings of $G$ with $%
\left\vert H\right\vert +\left\vert H^{\prime }\right\vert =\nu
_{2}$ and let $e$ be the edge of $G$ whose $1$-subdivision led to
the graph $G^{\prime }$. We will consider three cases:

Case 1: $e$ lies on a $H\bigtriangleup H^{\prime }$ alternating
cycle $C$.

As $G$ is connected and is not an even cycle, there is a
vertex $v\in V(C)$ with $d_{G}(v)\geq 3$. Clearly, there is a vertex $%
u\notin V(C)$ with $d_{G}(u)=2$ and $(u,v)\notin H\cup H^{\prime }$. Let $%
(u,w)$ be the other ($\neq (u,v)$) edge incident to $u$ and $f$ be
an edge of $C$ incident to $v$. Note that since $v$ is incident to
two edges lying
on $C$ we, without loss of generality, may assume $f$ to be different from $%
e $. Let $P_{0}$ be a path in $G$ whose edge-set coincides with $%
E(C)\backslash \{f\}$ and which starts from the vertex $v$. Now,
assume $P$ to be a path obtained from $P_{0}$ by adding the edge
$(u,v)$ to it, and let $P^{\prime }$ be the path of $G^{\prime }$
corresponding to $P$ (that is, the path obtained from $P$ by the
$1$-subdivision of the edge $e$).

Now, consider a pair of edge-disjoint matchings
$(H_{0},H_{0}^{\prime })$ of $G^{\prime }$ obtained in the following
way:

\begin{itemize}
\item if $(u,w)\notin H$ then alternatively add the edges of $P^{\prime }$
to $H_{0}$ and $H_{0}^{\prime }$ beginning from $H_{0}$;

\item if $(u,w)\notin H^{\prime }$ then alternatively add the edges of $%
P^{\prime }$ to $H_{0}$ and $H_{0}^{\prime }$ beginning from
$H_{0}^{\prime } $.
\end{itemize}

Define a pair of edge-disjoint matchings $(H_{1},H_{1}^{\prime })$ of $%
G^{\prime }$ as follows:%
\begin{equation*}
H_{1}=(H\backslash E(C))\cup H_{0},H_{1}^{\prime }=(H^{\prime
}\backslash E(C))\cup H_{0}^{\prime }.
\end{equation*}

Clearly,
\begin{equation*}
\nu' _{2}\geq \left\vert H_{1}\right\vert +\left\vert H_{1}^{\prime
}\right\vert =1+\left\vert H\right\vert +\left\vert H^{\prime
}\right\vert =1+\nu _{2}.
\end{equation*}

Case 2: $e$ lies on a $H\bigtriangleup H^{\prime }$ alternating path
$P$.

Let $P^{\prime }$ be the path of $G^{\prime }$ corresponding to $P$
(that is, the path obtained from $P$ by the $1$-subdivision of the
edge $e$).
Consider a pair of edge-disjoint matchings $(H_{0},H_{0}^{\prime })$ of $%
G^{\prime }$ obtained in the following way: alternatively add the edges of $%
P^{\prime }$ to $H_{0}$ and $H_{0}^{\prime }$. Define:%
\begin{equation*}
H_{1}=(H\backslash E(P))\cup H_{0},H_{1}^{\prime }=(H^{\prime
}\backslash E(P))\cup H_{0}^{\prime }.
\end{equation*}

Clearly,
\begin{equation*}
\nu' _{2}\geq \left\vert H_{1}\right\vert +\left\vert H_{1}^{\prime
}\right\vert =1+\left\vert H\right\vert +\left\vert H^{\prime
}\right\vert =1+\nu _{2}.
\end{equation*}

Case 3: $e\notin H\cup H^{\prime }$.

Due to (c) there is $u\in V$ with $d_{G}(u)=2$, such that $e$ is
incident to $u$. Let $f$ be the other ($\neq e$) edge of $G$ that is
incident to $u$, and assume $e^{\prime }$ to be the edge of
$G^{\prime }$ that is incident to $u$ in $G^{\prime }$ and is
different from $f$. Now, add the edge $e^{\prime }$ to $H$ if
$f\notin H$, and to $H^{\prime }$ if $f\notin H^{\prime }$. Clearly,
we constructed a pair of edge-disjoint matchings of $G^{\prime }$,
which contains $1+\nu _{2}$ edges, therefore%
\begin{equation*}
\nu' _{2}\geq 1+\nu _{2}.
\end{equation*}

(2) Note that if $G$ is an odd cycle then $G^{\prime }$ is an even one and $%
\nu' _{2}=2+\nu _{2}$, therefore, taking into account (1) and (b),
it suffices to show that if $G$ is not a cycle then $\nu' _{2}\leq
1+\nu _{2}$.

Let $(H,H^{\prime })$ be a pair of edge-disjoint matchings of
$G^{\prime }$ with $\left\vert H\right\vert +\left\vert H^{\prime
}\right\vert =\nu' _{2}$ and let $v$ be the new vertex of $G^{\prime
}$, that is, assume $\left\{ v\right\} =V'\backslash V$. We need to
consider three cases:

Case 1: $H\cup H^{\prime }$ contains at most one edge incident to
the vertex $v$.

Note that%
\begin{eqnarray*}
\nu _{2} &\geq &\left\vert (H\cup H^{\prime })\cap E\right\vert \geq
\left\vert (H\cup H^{\prime })\cap E(G-e)\right\vert \geq \\
&\geq &\left\vert H\right\vert +\left\vert H^{\prime }\right\vert
-1=\nu' _{2}-1
\end{eqnarray*}

or%
\begin{equation*}
\nu' _{2}\leq 1+\nu _{2}.
\end{equation*}

Case 2: The vertex $v$ belongs to an alternating component of $%
H\bigtriangleup H^{\prime }$ which is a path $P_{v}^{\prime }$.

Let $P_{v}$ be a path of $G$ containing the edge $e$ and corresponding to $%
P_{v}^{\prime }$, that is, let $P_{v}^{\prime }$ be obtained from
$P_{v}$ by the $1$-subdivision of the edge $e$. Consider a pair of
edge-disjoint matchings $(H_{0},H_{0}^{\prime })$ of $G$ defined as
follows: alternatively
add the edges of $P_{v}$ to $H_{0}$ and $H_{0}^{\prime }$. Define:%
\begin{equation*}
H_{1}=(H\backslash E(P_{v}^{\prime }))\cup H_{0},H_{1}^{\prime
}=(H^{\prime }\backslash E(P_{v}^{\prime }))\cup H_{0}^{\prime }.
\end{equation*}

Note that $(H_{1},H_{1}^{\prime })$ is a pair of edge-disjoint matchings of $%
G$. Moreover,%
\begin{equation*}
\nu _{2}\geq \left\vert H_{1}\right\vert +\left\vert H_{1}^{\prime
}\right\vert =\left\vert H\right\vert +\left\vert H^{\prime
}\right\vert -1=\nu' _{2}-1
\end{equation*}

or%
\begin{equation*}
\nu' _{2}\leq 1+\nu _{2}.
\end{equation*}

Case 3: The vertex $v$ belongs to an alternating component of $%
H\bigtriangleup H^{\prime }$ which is a cycle $C_{v}^{\prime }$.

Let $C_{v}$ be a cycle of $G$ containing the edge $e$ and corresponding to $%
C_{v}^{\prime }$, that is, let $C_{v}^{\prime }$ be obtained from
$C_{v}$ by the $1$-subdivision of the edge $e$. As $G$ is not a
cycle, we imply that there is a vertex $w\in V(C_{v}^{\prime })$
with $d_{G^{\prime }}(w)\geq 3$. Clearly, there is a vertex
$w^{\prime }\in V'\backslash
V(C_{v}^{\prime })$ such that $d_{G^{\prime }}(w^{\prime })=2$ and $%
(w,w^{\prime })\in E'$. Let $g$ be the other ($\neq (w,w^{\prime
})$) edge of $G^{\prime }$ incident to $w^{\prime }$. Since $w$ is
incident
to two edges lying on $C_{v}^{\prime }$, we imply that there is an edge $f$ $%
\neq e$ such that $f$ is incident to $w$. Let $P_{0v}$ be a path of
$G$, whose set of edges coincides with $E(C_{v})\backslash \{f\}$
and starts from $w$. Now consider the path $P_{v}$ obtained from
$P_{0v}$ by adding the edge $(w,w^{\prime })$ to it.

Consider a pair of edge-disjoint matchings $(H_{0},H_{0}^{\prime })$
of $G$ defined as follows:

\begin{itemize}
\item if $g\notin H$ \ then alternatively add the edges of $P_{v}$ to $H_{0}$
and $H_{0}^{\prime }$ beginning from $H_{0}$;

\item if $g\notin H^{\prime }$ then alternatively add the edges of $P_{v}$
to $H_{0}$ and $H_{0}^{\prime }$ beginning from $H_{0}^{\prime }$.
\end{itemize}

Define%
\begin{equation*}
H_{1}=(H\backslash E(C_{v}^{\prime }))\cup H_{0},H_{1}^{\prime
}=(H^{\prime }\backslash E(C_{v}^{\prime }))\cup H_{0}^{\prime }.
\end{equation*}

Note that $(H_{1},H_{1}^{\prime })$ is a pair of edge-disjoint matchings of $%
G$. Moreover,%
\begin{equation*}
\nu _{2}\geq \left\vert H_{1}\right\vert +\left\vert H_{1}^{\prime
}\right\vert =\left\vert H\right\vert +\left\vert H^{\prime
}\right\vert -1=\nu' _{2}-1
\end{equation*}

or%
\begin{equation*}
\nu' _{2}\leq 1+\nu _{2}.
\end{equation*}

The proof of the lemma \ref{Edge Subdivision} is completed.
\end{proof}

\section{The lemma}

In this section we prove a lemma that presents some lower bounds for
our parameters while we consider various subdivisions of graphs. The
aim of this lemma is the preparation of adequate theoretical tools
for understanding the growth of our parameters depending on the
numbers that the edges of graphs are subdivided. In contrast with
the proofs of the statements (a), (b), (c), (h), (i), that do not
include any induction, the proofs of the others significantly rely
on induction. Moreover, the basic tools for proving these statements
by induction are the proposition \ref{CubicPseudoGraph} and the
"loop-cut", the operation that helps us to reduce the number of
loops in a pseudo-graph. To understand the dynamics of the growth of
our parameters, we heavily use the lemma \ref{Edge Subdivision}.

Before we move on, we would like to define a class of graphs which
will play a crucial role in the proof of the main result of the
paper.

If $G_{0}$ is a cubic pseudo-graph such that the removal (not cut)
of its loops leaves a tree (if we adopt the convention presented in
\cite{Harary}, then we may say that the "underlying graph" of
$G_{0}$ is a tree; the simplest example of such a cubic pseudo-graph
is one from figure \ref{TrivialCase}), then consider the graph $G$
obtained from $G_{0}$ by $l(e)$-subdividing each edge $e$ of
$G_{0}$, where
\begin{equation*}
l(e)=\left\{
\begin{array}{ll}
1, & \text{if }e\text{ is a loop,} \\
2, & \text{otherwise.}%
\end{array}%
\right.
\end{equation*}%
Define $\mathfrak{M}$ to be the class of all those graphs $G$ that
can be
obtained in the mentioned way. Note that the members of the class $\mathfrak{%
M}$ are connected graphs.

\begin{lemma}
\label{PseudoGraphSubdivision} Let $G_{0}$ be a connected cubic
pseudo-graph, and consider the graph $G$ obtained from $G_{0}$ by $k(d)$%
-subdividing each edge $d$ of $G_{0}$, $k(d)\geq 1$. Suppose that,
for every edge $d$ of $G_{0}$, which is not a loop, we have:
$k(d)\geq 2$. Then:

\begin{description}
\item[(a)] If $G_{0}$ does not contain a loop then

\begin{description}
\item[(a1)] $\nu _{2}\geq \frac{7}{8}n ;$

\item[(a2)] $n \geq 4n_{0}
;$
\end{description}

\item[(b)] If $G_{0}$ contains an edge $f$ which is adjacent to two loops $e$
and $g$, then $G_{0}$ is the cubic pseudo-graph from figure
\ref{TrivialCase} and%
\begin{equation*}
\frac{\nu _{2}}{n }=\frac{k(e)+k(f)+k(g)+1}{%
k(e)+k(f)+k(g)+2}\text{;}
\end{equation*}

\item[(c)] If $G_{0}$ contains a loop $e$, then consider the cubic
pseudo-graph $G_{0}^{\prime }$ obtained from $G_{0}$ by cutting the
loop $e$ and the graph $G^{\prime }$ obtained from $G_{0}^{\prime }$
by $k^{\prime }(d^{\prime })$-subdividing each edge $d^{\prime }$ of
$G_{0}^{\prime }$, where
\begin{equation}
k^{\prime }(d^{\prime })=\left\{
\begin{array}{ll}
k(h)+k(h^{\prime })-2 & \text{if }d^{\prime }=g\text{,} \\
k(d^{\prime }) & \text{otherwise.}%
\end{array}%
\right.  \label{KPrimeDefinition}
\end{equation}%
Then:

\begin{description}
\item[(c1)] $n_0 =n_{0}^{\prime}+2;$

\item[(c2)] $n =n'+k(f)+k(e)+4;$

\item[(c3)] $\nu _{1}\geq \nu' _{1}+\left[ \frac{k(f)}{2}%
\right] +\left[ \frac{k(e)+1}{2}\right] +1;$

\item[(c4)] $\nu _{2}\geq \nu' _{2}+k(f)+k(e)+3;$
\end{description}

\item[(d)]

\begin{description}
\item[(d1)] $\nu _{2}\geq \frac{5}{6}n ;$

\item[(d2)] $n \geq 3n_0;$
\end{description}

\item[(e)]

\begin{description}
\item[(e1)] If $G_{0}$ contains a loop $e$ such that $k(e)\geq 2$ then $\nu
_{2}\geq \frac{6}{7}n $ and $n \geq \frac{7}{2}n_0;$

\item[(e2)] If $G_{0}$ contains an edge $f$ such that $f$ is not a loop and $%
k(f)\geq 3$ then $\nu _{2}\geq \frac{6}{7}n$ and $n \geq
\frac{7}{2}n_0;$
\end{description}

\item[(f)] $\nu _{1}\geq \frac{3}{7}n ;$

\item[(g)] If $G\in \mathfrak{M}$ then $\nu _{1}\geq \frac{6}{13}n ;$

\item[(h)] If a cubic pseudo-graph $G_{0}^{\prime }$ is obtained from $G_{0}$
by cutting its loop $e$ and if a graph $G^{\prime }$ is obtained from $%
G_{0}^{\prime }$ by $k^{\prime }(d^{\prime })$-subdividing each edge $%
d^{\prime }$ of $G_{0}^{\prime }$, where $k^{\prime }(d^{\prime })$
is defined according to (\ref{KPrimeDefinition}), then if $n' \geq
\frac{7}{2}n'_0 $ then $n \geq \frac{7}{2}n_0$; in other words, the
property $n <\frac{7}{2}n_0 $ is an invariant
for the operation of cutting a loop and defining $k^{\prime }$ according to (%
\ref{KPrimeDefinition});

\item[(i)] If $n <\frac{7}{2}n_0 $ then $G\in \mathfrak{M}$.
\end{description}
\end{lemma}

\begin{proof}
(a) For the proof of (a1) consider a graph $G^{\prime }$ obtained from $%
G_{0} $ by $1$-subdividing each edge of $G_{0}$. Note that
$G^{\prime }$ satisfies the conditions of (\ref{NuIN(2,k)graphs}) of
the lemma \ref{Bipartite 2->=3}, thus (see
the equality (\ref{eq3}))%
\begin{equation*}
\nu' _{2}=\frac{4}{5}n' =\frac{%
4}{5}(n_0 +m_0)=\frac{%
4}{5}\cdot \frac{5}{2}\cdot n_0=2n_0
\end{equation*}%
therefore due to lemma \ref{Edge Subdivision} we have:%
\begin{equation}
\frac{\nu _{2}}{n }=\frac{\nu' _{2}+\sum_{e\in E_{0}}(k(e)-1)}{n'
+\sum_{e\in E_{0}}(k(e)-1)}\text{.} \label{asterik}
\end{equation}

Note that for each $e\in E_{0}$ $k(e)\geq 2$, hence
\begin{equation*}
\sum_{e\in E_{0}}(k(e)-1)\geq m_0 =%
\frac{3}{2}n_0.
\end{equation*}%
Taking into account (\ref{asterik}) we get:
\begin{equation*}
\frac{\nu _{2}}{n }=\frac{2n_0 +\sum_{e\in E_{0}}(k(e)-1)}{\frac{5}{2}%
n_0 +\sum_{e\in E_{0}}(k(e)-1)}\geq \frac{2n_0
+\frac{3}{2}n_0 }{\frac{5}{2}n_0 +\frac{3}{2}%
n_0 }=\frac{7}{8},
\end{equation*}%
thus
\begin{equation*}
\nu _{2}\geq \frac{7}{8}n .
\end{equation*}

For the proof of (a2) let us note that as $G_{0}$ does not contain a
loop, for each edge $f$ of $G_{0}$ we have $k(f)\geq 2$, thus
\begin{equation*}
n =n_0 +\sum_{f\in E_{0}}k(f)\geq n_0+2m_0 =4n_0.
\end{equation*}

(b) Note that
\begin{equation*}
n =n_0 +k(e)+k(f)+k(g)=2+k(e)+k(f)+k(g).
\end{equation*}
Since $f$ is not a loop, we have $k(f)\geq 2$ thus
\begin{equation*}
\nu _{2}=m -2=1+k(e)+k(f)+k(g),
\end{equation*}
and
\begin{equation*}
\frac{\nu _{2}}{n }=\frac{k(e)+k(f)+k(g)+1}{%
k(e)+k(f)+k(g)+2}.
\end{equation*}

(c) The proof of (c1) follows directly from the definition of the
operation of cutting loops. For the proof of (c2) note that
\begin{eqnarray*}
n &=&n'
-k'(g)+k(h)+k(h')+1+k(f)+1+k(e)= \\
&=&n'+k(f)+k(e)+4
\end{eqnarray*}
since $k'(g)=k(h)+k(h')-2$ (see (\ref{KPrimeDefinition})).

For the proof of (c3) and (c4) let us introduce some additional
notations. Let $C_{e},P_{f},P_{h},P_{h'}$ be the cycle and paths of
$G$
corresponding to the edges $e,f,h,h'$ of the cubic pseudo-graph $%
G_{0}$. Let $K_{g}$ be the cycle or a path of $G'$ corresponding to
the edge $g$ of the cubic pseudo-graph $G'_{0}$.

Let $F'$ be a maximum matching of the graph $G'$. Define
$\varepsilon =\varepsilon (F')$ as the number of vertices from
$\{u,v\}$ which
are saturated by an edge from $F'\cap E(K_{g})$. Note that if $%
u\neq v$ then $0\leq \varepsilon \leq 2$ and if $u=v$ then $%
0\leq \varepsilon \leq 1$.

Consider a subset of edges of the graph $G$ defined as:
\begin{equation*}
F=(F'\backslash E(K_{g}))\cup F_{h,h'}\cup F_{f}\cup F_{e}%
\,
\end{equation*}%
where $F_{h,h'}$ is a maximum matching of a path $P_{h,h'}$ obtained
from the paths $P_{h}$ and $P_{h'}$ as follows:
\begin{equation*}
P_{h,h'}=\left\{
\begin{array}{ll}
P_{h}\backslash \{u,v_{0}\},v_{0},P_{h^{\prime }}\backslash
\{v_{0},v\} &
\text{if }\varepsilon =0\text{;} \\
P_{h}\backslash \{v_{0}\},v_{0},P_{h^{\prime }}\backslash \{v_{0}\} & \text{%
if }\varepsilon =2\text{;} \\
P_{h}\backslash \{v_{0}\},v_{0},P_{h^{\prime }}\backslash
\{v_{0},v\} &
\begin{array}{l}
\text{if }\varepsilon =1\text{ and an edge } \\
\text{of }F^{\prime }\cap E(K_{g})\text{ saturates }u\text{;}%
\end{array}
\\
P_{h}\backslash \{u,v_{0}\},v_{0},P_{h^{\prime }}\backslash
\{v_{0}\} &
\begin{array}{l}
\text{if }\varepsilon =1\text{ and an edge } \\
\text{of }F^{\prime }\cap E(K_{g})\text{ saturates }v\text{;}%
\end{array}%
\end{array}%
\right.
\end{equation*}%
$F_{f}$ is a maximum matching of $P_{f}\backslash \{u_{0},v_{0}\}$, and $%
F_{e}$ is a maximum matching of $C_{e}$.

Note that if $u=v$ and $\varepsilon =1$ then we define the path
$P_{h,h^{\prime }}$ in two ways. We would like to stress that our
results do not depend on the way the path $P_{h,h^{\prime }}$ is
defined.

By the construction of $F$, $F$ is a matching of $G$. Moreover,%
\begin{gather*}
\nu _{1}\geq \left\vert F\right\vert =\left\vert F^{\prime
}\right\vert -\left\vert F^{\prime }\cap E(K_{g})\right\vert
+\left\vert F_{h,h^{\prime
}}\right\vert +\left\vert F_{f}\right\vert +\left\vert F_{e}\right\vert = \\
=\nu' _{1}-\left[ \frac{k^{\prime }(g)+\varepsilon %
}{2}\right] +\left[ \frac{k(h)+k(h^{\prime })+1+\varepsilon }{2}%
\right] +\left[ \frac{k(f)}{2}\right] + \\
+\left[ \frac{k(e)+1}{2}\right] =\nu' _{1}-\left[ \frac{%
k(h)+k(h^{\prime })+\varepsilon }{2}\right] +1+ \\
+\left[ \frac{k(h)+k(h^{\prime })+1+\varepsilon }{2}\right] +%
\left[ \frac{k(f)}{2}\right] +\left[ \frac{k(e)+1}{2}\right] \geq \\
\geq \nu' _{1}+\left[ \frac{k(f)}{2}\right] +\left[ \frac{k(e)+1%
}{2}\right] +1
\end{gather*}%
as%
\begin{equation*}
\left[ \frac{k(h)+k(h^{\prime })+1+\varepsilon }{2}\right] \geq %
\left[ \frac{k(h)+k(h^{\prime })+\varepsilon }{2}\right] \text{.%
}
\end{equation*}

Now, let us turn to the proof of (c4). Let $(H_{1}^{\prime
},H_{2}^{\prime
}) $ be a pair of edge-disjoint matchings of $G^{\prime }$ such that $%
\left\vert H_{1}^{\prime }\right\vert +\left\vert H_{2}^{\prime
}\right\vert =\nu' _{2}$. Define $\delta =\delta (H_{1}^{\prime
},H_{2}^{\prime })$ as the number of vertices from $\{u,v\}$ which
are saturated by an edge from $(H_{1}^{\prime }\cup H_{2}^{\prime
})\cap E(K_{g})$. Note that if $u\neq v$ then $0\leq \delta \leq 2$
and if $u=v$ then $0\leq \delta \leq 1$. We need to consider two
cases:

Case 1: $0\leq \delta \leq 1$;

Define a pair of edge-disjoint matchings $\left( H_{1},H_{2}\right)
$ of $G$ as follows:
\begin{eqnarray*}
H_{1} &=&(H_{1}^{\prime }\backslash E(K_{g}))\cup H_{1hh^{\prime
}}\cup
H_{1fe}\text{,} \\
H_{2} &=&(H_{2}^{\prime }\backslash E(K_{g}))\cup H_{2hh^{\prime
}}\cup H_{2fe}\text{,}
\end{eqnarray*}%
where $H_{1hh^{\prime }}$,$H_{2hh^{\prime }}$ are obtained from a path $%
P_{hh^{\prime }}$ alternatively adding its edges to $H_{1hh^{\prime }}$ and $%
H_{2hh^{\prime }}$; $H_{1fe}$,$H_{2fe}$ are obtained from a path
$P_{fe}$
alternatively adding its edges to $H_{1fe}$ and $H_{2fe}$, and the paths $%
P_{hh^{\prime }}$ and $P_{fe}$ are defined as
\begin{equation*}
P_{h,h^{\prime }}=\left\{
\begin{array}{ll}
P_{h}\backslash \{u,v_{0}\},v_{0},P_{h^{\prime }}\backslash
\{v_{0},v\} &
\text{ if }\delta =0\text{;} \\
P_{h}\backslash \{v_{0}\},v_{0},P_{h^{\prime }}\backslash
\{v_{0},v\} &
\begin{array}{l}
\text{if }\delta =1\text{ and an edge } \\
\text{of }(H_{1}^{\prime }\cup H_{2}^{\prime })\cap E(K_{g})\text{
saturates
}u\text{;}%
\end{array}
\\
P_{h}\backslash \{u,v_{0}\},v_{0},P_{h^{\prime }}\backslash
\{v_{0}\} &
\begin{array}{l}
\text{if }\delta =1\text{ and an edge } \\
\text{of }(H_{1}^{\prime }\cup H_{2}^{\prime })\cap E(K_{g})\text{
saturates
}v\text{;}%
\end{array}%
\end{array}%
\right.
\end{equation*}%
\begin{equation*}
P_{fe}=P_{f}\backslash \{v_{0},u_{0}\},u_{0},C_{e}\backslash \{u_{0}\}\text{.%
}
\end{equation*}%
Again, let us note that if $u=v$ and $\delta =1$ then we define the
path $P_{h,h^{\prime }}$ in two ways. We would like to stress that
our results do not depend on the way the path $P_{h,h^{\prime }}$ is
defined.

Note that
\begin{gather*}
\nu _{2}\geq \left\vert H_{1}\right\vert +\left\vert
H_{2}\right\vert =\left\vert (H_{1}^{\prime }\cup H_{2}^{\prime
})\backslash E(K_{g})\right\vert +(\left\vert H_{1hh^{\prime
}}\right\vert +\left\vert
H_{2hh^{\prime }}\right\vert )+ \\
+(\left\vert H_{1fe}\right\vert +\left\vert H_{2fe}\right\vert
)=\left\vert H_{1}^{\prime }\right\vert +\left\vert H_{2}^{\prime
}\right\vert -\left\vert (H_{1}^{\prime }\cup H_{2}^{\prime })\cap
E(K_{g})\right\vert +
\\
+\left\vert E(P_{hh^{\prime }})\right\vert +\left\vert
E(P_{fe})\right\vert \geq \nu' _{2}-((k^{\prime }(g)+\delta
)-1)+ \\
+((k(h)+k(h^{\prime })+\delta +1)-1)+((k(f)+k(e)+1)-1)= \\
=\nu' _{2}-(k(h)+k(h^{\prime })+\delta -3)+(k(h)+k(h^{\prime
})+\delta )+ \\
+(k(f)+k(e))=\nu' _{2}+k(f)+k(e)+3\text{.}
\end{gather*}

Case 2: $\delta =2$;

Define a pair of edge-disjoint matchings $\left( H_{1},H_{2}\right)
$ of $G$ as follows:
\begin{eqnarray*}
H_{1} &=&(H_{1}^{\prime }\backslash E(K_{g}))\cup H_{1hfe}\cup
H_{1h^{\prime
}}\text{,} \\
H_{2} &=&(H_{2}^{\prime }\backslash E(K_{g}))\cup H_{2hfe}\cup
H_{2h^{\prime }}\text{,}
\end{eqnarray*}%
where $H_{1hfe}$,$H_{2hfe}$ are obtained from a path $P_{hfe}$
alternatively
adding its edges to $H_{1hfe}$ and $H_{2hfe}$; $H_{1h^{\prime }}$,$%
H_{2h^{\prime }}$ are obtained from the path $P_{h^{\prime
}}\backslash
\{v_{0}\}$ alternatively adding its edges to $H_{1h^{\prime }}$ and $%
H_{2h^{\prime }}$, and the path $P_{hfe}$ is defined as
\begin{equation*}
P_{hfe}=P_{h}\backslash \{v_{0}\},v_{0},P_{f}\backslash
\{v_{0},u_{0}\},u_{0},C_{e}\backslash \{u_{0}\}\text{.}
\end{equation*}%
Note that%
\begin{gather*}
\nu _{2}\geq \left\vert H_{1}\right\vert +\left\vert
H_{2}\right\vert =\left\vert (H_{1}^{\prime }\cup H_{2}^{\prime
})\backslash E(K_{g})\right\vert +(\left\vert H_{1hfe}\right\vert
+\left\vert
H_{2hfe}\right\vert )+ \\
+(\left\vert H_{1h^{\prime }}\right\vert +\left\vert H_{2h^{\prime
}}\right\vert )=\left\vert H_{1}^{\prime }\right\vert +\left\vert
H_{2}^{\prime }\right\vert -\left\vert (H_{1}^{\prime }\cup
H_{2}^{\prime
})\cap E(K_{g})\right\vert + \\
+\left\vert E(P_{hfe})\right\vert +\left\vert E(P_{h^{\prime
}}\backslash
\{v_{0}\})\right\vert \geq \nu' _{2}-((k^{\prime }(g)+2)-1)+ \\
+(1+k(h)+1+k(f)+1+k(e)-1)+((k(h^{\prime })+1)-1)= \\
=\nu' _{2}-(k(h)+k(h^{\prime
})-1)+(k(h)+k(f)+k(e)+2)+k(h^{\prime })= \\
=\nu' _{2}+k(f)+k(e)+3\text{.}
\end{gather*}

(d) We will give a simultaneous proof of the statements (d1) and
(d2). Note that if $G_{0}$ does not contain a loop then (a1) and
(a2) imply that
\begin{equation*}
\nu _{2}\geq \frac{7}{8}n
>\frac{5}{6}n \text{, and }n \geq 4n_0 >3n_0,
\end{equation*}%
thus without loss of generality, we may assume that $G_{0}$ contains
a loop. Our proof is by induction on $n_0$. Clearly, if $n_0=2$ then
$G_{0}$ is the pseudo-graph from figure \ref{TrivialCase}, thus (b)
implies that
\begin{equation*}
\frac{\nu _{2}}{n }\geq \frac{5}{6},\text{ and }%
n =2+k(e)+k(f)+k(g)\geq 6=3n_0
\end{equation*}

as $k(e),k(g)\geq 1$ and $k(f)\geq 2$. Note that $\nu
_{2}=\frac{5}{6}n$ or $n =3n_0$ if $k(e)=k(g)=1$ and $k(f)=2$.

Now, by induction, assume that for every graph $G^{\prime }$
obtained from a cubic pseudo-graph $G_{0}^{\prime }$ ($n'_0 <n_0$)
by $k^{\prime }(e^{\prime }) $-subdividing each edge $e^{\prime }$
of $G_{0}^{\prime }$, we have
\begin{equation*}
\nu' _{2}\geq \frac{5}{6}n' \text{ and }n' \geq 3n'_0,
\end{equation*}%
and consider the cubic pseudo-graph $G_{0}$ ($n_0 \geq 4$) and its
corresponding graph $G$.

Let $e$ be a loop of $G_{0}$, and consider a cubic pseudo-graph $%
G_{0}^{\prime },$ obtained from $G_{0},$ by cutting the loop $e$
((a) of figure \ref{loopcut}). Note that $G_{0}^{\prime }$ is
well-defined, since $n_0 \geq 4$. As $n'_0 < n_0,$ due to induction
hypothesis, we have
\begin{equation}
\nu' _{2}\geq \frac{5}{6}n' \text{ and }n' \geq 3n'_0 \text{,}
\label{InductionBound}
\end{equation}%
where $G^{\prime }$ is obtained from $G_{0}^{\prime }$ by $k^{\prime
}(d^{\prime })$-subdividing each edge $d^{\prime }$ of
$G_{0}^{\prime }$,
and the mapping $k^{\prime }$ is defined according to (\ref{KPrimeDefinition}%
). On the other hand, due to (c1), (c2) and (c4), we have%
\begin{gather*}
n_0 =n'_0 +2;
\\
n =n' +k(f)+k(e)+4%
\text{,} \\
\nu _{2}\geq \nu' _{2}+k(f)+k(e)+3.
\end{gather*}%
Since $k(f)\geq 2$, $k(e)\geq 1$ we have
\begin{gather*}
\frac{k(f)+k(e)+3}{k(f)+k(e)+4}\geq \frac{6}{7}>\frac{5}{6}\text{, and} \\
\frac{k(f)+k(e)+4}{2}\geq \frac{7}{2}>3
\end{gather*}%
and therefore due to (\ref{InductionBound}) and proposition \ref%
{FractionInequality}, we get:
\begin{gather*}
\frac{\nu _{2}}{n}\geq \frac{\nu' _{2}+k(f)+k(e)+3}{n'+k(f)+k(e)+4}\geq \frac{%
5}{6}\text{, and} \\
\frac{n}{n_0}=\frac{%
n'+k(f)+k(e)+4}{n'_0+2}\geq 3.
\end{gather*}

(e) We will prove (e1) by induction on $n_0$. Note that if $n_0=2,$
then $G_{0}$ is the pseudo-graph from figure \ref{TrivialCase}, thus%
\begin{equation*}
n =k(e)+k(f)+k(g)+2=\frac{k(e)+k(f)+k(g)+2}{2}\cdot n_0
\end{equation*}%
and due to (b)%
\begin{equation*}
\frac{\nu _{2}}{n}=\frac{k(e)+k(f)+k(g)+1}{%
k(e)+k(f)+k(g)+2}.
\end{equation*}%
Now if $G_{0}$ satisfies (e1), then taking into account that $k(g)\geq 1$, $%
k(e)\geq 1$, $\max \{k(e),k(g)\}\geq 2$ and $k(f)\geq 2$, we get $%
k(e)+k(f)+k(g)\geq 5$, and therefore
\begin{equation*}
\frac{\nu _{2}}{n}\geq \frac{6}{7}\text{ and }%
n \geq \frac{7}{2}n_0.
\end{equation*}

Now, by induction, assume that for every graph $G^{\prime },$
obtained from a cubic pseudo-graph $G_{0}^{\prime }$ ($n'_0<n_0$),
by $k^{\prime }(e^{\prime }) $-subdividing each edge $e^{\prime }$
of $G_{0}^{\prime }$, we have
\begin{equation*}
\nu' _{2}\geq \frac{6}{7}n' \text{ and }n'\geq \frac{7}{2}n'_0,
\end{equation*}%
provided that $G_{0}^{\prime }$ satisfies (e1), and consider the
cubic pseudo-graph $G_{0}$ ($n_0 \geq 4$) and its corresponding
graph $G$. We need to consider two cases:

Case 1: $G_{0}$ contains at least two loops.

Let $e_{0}$ be a loop of $G_{0}$ that differs from $e$. Consider the
cubic pseudo-graph $G_{0}^{\prime },$ obtained from $G_{0},$ by
cutting the loop $e_{0}$ ((a) of figure \ref{loopcut}), and the graph $%
G^{\prime },$ obtained from a cubic pseudo-graph $G_{0}^{\prime },$ by $%
k^{\prime }(e^{\prime })$-subdividing each edge $e^{\prime }$ of $%
G_{0}^{\prime }$, where the mapping $k^{\prime }$ is defined according to (%
\ref{KPrimeDefinition}).

Since $n'_0 <n_0$ and $e\in E_{0}^{\prime }$, due to induction
hypothesis, we have
\begin{equation*}
\nu' _{2}\geq \frac{6}{7}n' \text{ and }n' \geq \frac{7}{2}n'_0
\end{equation*}%
(c1), (c2) and (c4) imply that%
\begin{gather*}
n_0 =n'_0+2;
\\
n=n'+k(f)+k(e_{0})+4, \\
\nu _{2}\geq \nu _{2}+k(f)+k(e_{0})+3.
\end{gather*}%
Since $k(f)\geq 2$, $k(e_{0})\geq 1$ we have
\begin{gather*}
\frac{k(f)+k(e_{0})+3}{k(f)+k(e_{0})+4}\geq \frac{6}{7}\text{, and} \\
\frac{k(f)+k(e_{0})+4}{2}\geq \frac{7}{2}
\end{gather*}%
and therefore due to proposition \ref{FractionInequality}, we get:
\begin{gather*}
\frac{\nu _{2}}{n }\geq \frac{\nu' _{2}+k(f)+k(e_{0})+3}{n'+k(f)+k(e_{0})+4}%
\geq \frac{6}{7}\text{, and} \\
\frac{n}{n_0}=\frac{n'+k(f)+k(e_{0})+4}{n'_0+2}\geq \frac{7}{2}.
\end{gather*}

Case 2: $G_{0}$ contains exactly one loop.

Let $e-$the only loop of $G_{0}-$ be adjacent to the edge $d$. Let
$u_{0}$ be the vertex of $G_{0}$ that is incident to $d$ and $e$, and let $%
d=(u_{0},v_{0})$. Let $h$ and $h^{\prime }$ ($h\neq h^{\prime }$) be
two edges that differ from $d$ and are incident to $v_{0}$. Finally,
let $u$ and $v$ be the endpoints of $h$ and $h^{\prime }$ that are
not incident to $d$, respectively.

Subcase 2.1: $u\neq v$.

Consider a cubic pseudo-graph $G_{0}^{\prime }$ obtained from
$G_{0}$ by cutting the loop $e$ and the graph $G^{\prime }$ obtained
from a cubic pseudo-graph $G_{0}^{\prime }$ by $k^{\prime
}(e^{\prime })$-subdividing each edge $e^{\prime }$ of
$G_{0}^{\prime }$, where the mapping $k^{\prime }$ is defined
according to (\ref{KPrimeDefinition}). As $G_{0}^{\prime }$ does
not contain a loop, due to (a1) and (a2), we have%
\begin{equation}
\nu' _{2}\geq \frac{7}{8}n' \text{ and }n' \geq 4n'_0.
\label{7-8Bound}
\end{equation}%
(c1), (c2) and (c4) imply that%
\begin{gather*}
n_0 =n'_0+2;
\\
n =n'+k(d)+k(e)+4%
\text{,} \\
\nu _{2}\geq \nu' _{2}+k(d)+k(e)+3.
\end{gather*}%
Since $k(e)\geq 2$, $k(d)\geq 2$ we have
\begin{equation*}
k(e)+k(d)\geq 4\text{,}
\end{equation*}%
thus
\begin{gather*}
\frac{k(d)+k(e)+3}{k(d)+k(e)+4}\geq \frac{7}{8}>\frac{6}{7}\text{, and} \\
\frac{k(d)+k(e)+4}{2}\geq 4>\frac{7}{2}\text{.}
\end{gather*}%
Due to (\ref{7-8Bound}) and proposition \ref{FractionInequality}, we
get:
\begin{gather*}
\frac{\nu _{2}}{n }\geq \frac{\nu' _{2}+k(d)+k(e)+3}{n'+k(d)+k(e)+4}\geq \frac{%
6}{7}\text{, and} \\
\frac{n}{n_0}=\frac{%
n'+k(d)+k(e)+4}{n'_0+2}\geq \frac{7}{2}.
\end{gather*}

Subcase 2.2: $u=v$.

Let $h^{\prime \prime }$ be the edge which is incident to $u$ and is
different from $h$ and $h^{\prime }$, and let $h^{\prime \prime
}=(u,w)$ (figure \ref{Reduction}).

\begin{figure}[h]
\begin{center}
\includegraphics[width=24pc]{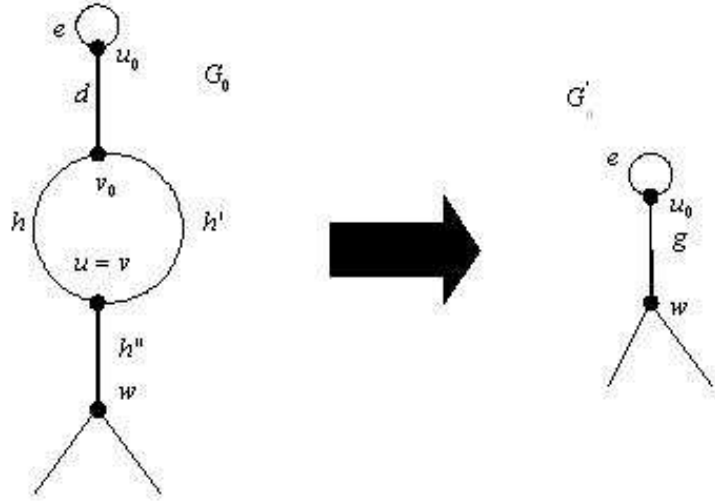}\\
\caption{Reducing $G_{0}$ to $%
G'_{0}$}\label{Reduction}
\end{center}
\end{figure}

Define a cubic pseudo-graph $G'_{0}$ as follows:%
\begin{eqnarray*}
G'_{0}&=&(G_{0}\backslash \{v_{0},u\})\cup \{g\}\text{, where} \\
g &=&(u_{0},w)\text{,}
\end{eqnarray*}%
and consider the graph $G'$ obtained from $G'_{0}$ by $%
k'(e')$-subdividing each edge $e'$ of $%
G'_{0}$, where
\begin{equation*}
k'(e')=\left\{
\begin{array}{ll}
k(d)+k(h^{\prime \prime })-2 & \text{if }e^{\prime }=g\text{,} \\
k(e^{\prime }) & \text{otherwise.}%
\end{array}%
\right.
\end{equation*}%
Note that $e\in E_{0}^{\prime }$, $n'_0 <n_0$ and $k^{\prime
}(e)=k(e)\geq
2$ thus, due to induction hypothesis, we have:%
\begin{equation}
\nu' _{2}\geq \frac{6}{7}n' \text{ and }n' \geq \frac{7}{2}n'_0.
\label{Induction6-7}
\end{equation}

It is not hard to see that
\begin{gather*}
n_0=n'_0+2;
\\
n =n'+k(h)+k(h^{\prime })+4, \\
\nu _{2}\geq \nu' _{2}+k(h)+k(h^{\prime })+3.
\end{gather*}
As $k(h),k(h^{\prime })\geq 2$, we have
\begin{eqnarray*}
\frac{k(h)+k(h^{\prime })+3}{k(h)+k(h^{\prime })+4} &\geq &\frac{7}{8}>\frac{%
6}{7}\text{, and} \\
\frac{k(h)+k(h^{\prime })+4}{2} &\geq &4>\frac{7}{2},
\end{eqnarray*}%
therefore due to (\ref{Induction6-7}) and proposition \ref%
{FractionInequality}, we get:%
\begin{gather*}
\frac{\nu _{2}}{n}\geq \frac{\nu' _{2}+k(h)+k(h^{\prime })+3}{n'
+k(h)+k(h^{\prime })+4}\geq \frac{6}{7}\text{, and} \\
\frac{n }{n_0}=\frac{%
n'+k(h)+k(h^{\prime })+4}{n'_0+2}\geq \frac{7}{2}.
\end{gather*}%
The proof of (e1) is completed. Now, let us turn to the proof of
(e2). Note that if $G_{0}$ does not contain a loop then (a1) and
(a2) imply that
\begin{equation*}
\nu _{2}\geq \frac{7}{8}n >\frac{6}{7}n \text{, and }n \geq 4n_0
>\frac{7}{2}n_0,
\end{equation*}%
thus, without loss of generality, we may assume that $G_{0}$
contains a loop.
Our proof is by induction on $n_0$. Clearly, if $%
n_0=2$ then $G_{0}$ is the pseudo-graph from figure
\ref{TrivialCase},
\begin{equation*}
n=k(e)+k(f)+k(g)+2=\frac{k(e)+k(f)+k(g)+2}{2}\cdot n_0
\end{equation*}%
and due to (b)%
\begin{equation*}
\frac{\nu _{2}}{n}=\frac{k(e)+k(f)+k(g)+1}{%
k(e)+k(f)+k(g)+2}\text{.}
\end{equation*}%
Now, if $G_{0}$ satisfies (e2) then $k(f)\geq 3$ and taking into
account that $k(g)\geq 1$, $k(e)\geq 1$, we get $k(e)+k(f)+k(g)\geq
5$, therefore
\begin{equation*}
\frac{\nu _{2}}{n}\geq \frac{6}{7}\text{ and }%
n \geq \frac{7}{2}n_0.
\end{equation*}

Now, by induction, assume that for every graph $G^{\prime }$
obtained from a cubic pseudo-graph $G_{0}^{\prime }$ ($n'_0 <n_0$)
by $k^{\prime }(e^{\prime }) $-subdividing each edge $e^{\prime }$
of $G_{0}^{\prime }$, we have
\begin{equation*}
\nu' _{2}\geq \frac{6}{7}n' \text{ and }n' \geq \frac{7}{2}n'_0
\end{equation*}%
and consider the cubic pseudo-graph $G_{0}$ ($n_0\geq 4$) and its
corresponding graph $G$.

Case 1: There is an edge $f^{\prime }=(u_{0},v_{0})$ such that $f$ and $%
f^{\prime }$ form a cycle of the length two (figure \ref{Case of
multiple edge})

\begin{figure}[h]
\begin{center}
\includegraphics[height=15pc]{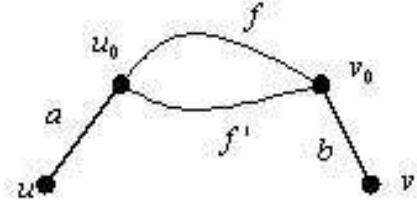}\\
\caption{The case of multiple edge}\label{Case of multiple edge}
\end{center}
\end{figure}

Let $a,b,f,f^{\prime },u_{0},v_{0},u,v$ be the edges and vertices as
on
figure \ref{Case of multiple edge}. Consider a cubic pseudo-graph $%
G_{0}^{\prime }$, defined as follows:%
\begin{eqnarray*}
G_{0}^{\prime } &=&(G_{0}\backslash \{u_{0},v_{0}\})\cup
\{g\}\text{, where}
\\
g &=&(u,v)\text{,}
\end{eqnarray*}%
and consider the graph $G^{\prime }$ obtained from $G_{0}^{\prime }$ by $%
k^{\prime }(e^{\prime })$-subdividing each edge $e^{\prime }$ of $%
G_{0}^{\prime }$, where
\begin{equation*}
k^{\prime }(e^{\prime })=\left\{
\begin{array}{ll}
k(f) & \text{if }e^{\prime }=g\text{,} \\
k(e^{\prime }) & \text{otherwise.}%
\end{array}%
\right.
\end{equation*}%
Note that
\begin{gather*}
n_0=n'_0+2;
\\
n =n'+k(a)+k(b)+k(f^{\prime })+2\text{,} \\
\nu _{2}\geq \nu' _{2}-(k(f)+1)+k(a)+k(b)+k(f^{\prime
})+2+1+k(f)-1= \\
=\nu' _{2}+k(a)+k(b)+k(f^{\prime })+1.
\end{gather*}%
Let us show that
\begin{equation*}
\nu' _{2}\geq \frac{6}{7}n' \text{ and }n'\geq \frac{7}{2}n'_0.
\end{equation*}%
First of all note that $n'_0 <n_0$ and $k^{\prime }(g)=k(f)\geq 3,$
therefore if $g$ is not a loop of $G_{0}^{\prime }$ ($u\neq v$) then
the inequalities follow directly from the induction hypothesis. On
the other hand, if $g$ is a loop of $G_{0}^{\prime }$ ($u=v$) then
the same inequalities hold due to (e1).

Since
\begin{gather*}
\frac{k(a)+k(b)+k(f^{\prime })+1}{k(a)+k(b)+k(f^{\prime })+2}\geq \frac{7}{8}%
>\frac{6}{7}\text{, and} \\
\frac{k(a)+k(b)+k(f^{\prime })+2}{2}\geq 4>\frac{7}{2}\text{.}
\end{gather*}%
proposition \ref{FractionInequality} implies that
\begin{gather*}
\frac{\nu _{2}}{n }\geq \frac{\nu' _{2}+k(a)+k(b)+k(f^{\prime
})+1}{n'+k(a)+k(b)+k(f^{\prime })+2}\geq \frac{6}{7}\text{, and} \\
\frac{n}{n_0}=\frac{n'+k(a)+k(b)+k(f^{\prime })+2}{n'_0+2}\geq
\frac{7}{2}.
\end{gather*}

Case 2: $G_{0}$ contains at least two loops and does not satisfy the
condition of the case 1.

As $G_{0}$ is connected and $n_0\geq 4$,
there is a loop $e$ of $G_{0}$ such that $e$ is not adjacent to $%
f $. Let $d$ be the edge adjacent to the edge $e$. Let $u_{0}$ be
the vertex
of $G_{0}$ that is incident to $d$ and $e$, and let $d=(u_{0},v_{0})$. Let $%
h $ and $h^{\prime }$ be two edges that differ from $d$ and are incident to $%
v_{0}$. Finally, let $u$ and $v$ be the endpoints of $h$ and
$h^{\prime }$ that are not incident to $d$, respectively.

Consider the cubic pseudo-graph $G_{0}^{\prime }$ obtained from
$G_{0}$ by cutting the loop $e$ and the graph $G^{\prime }$ obtained
from a cubic pseudo-graph $G_{0}^{\prime }$ by $k^{\prime
}(e^{\prime })$-subdividing each edge $e^{\prime }$ of
$G_{0}^{\prime }$, where the mapping $k^{\prime }$ is defined
according to (\ref{KPrimeDefinition}). Note that $n'_0 < n_0$.

Let us show that $G_{0}^{\prime }$ satisfies the condition of (e2).
Clearly, if $f\in E_{0}^{\prime }$ then we are done, thus we may assume that $%
f\notin E_{0}^{\prime }$. Since $d\neq f$, we imply that $f\in
\{h,h^{\prime }\}$. As $G_{0}$ does not satisfy the condition of the
case 1, the edge $g\in E_{0}^{\prime }$ is not a loop of $%
G_{0}^{\prime }$ and
\begin{equation*}
k^{\prime }(g)=k(h)+k(h^{\prime })-2\geq 3\text{.}
\end{equation*}
Thus $G_{0}^{\prime }$ satisfies the condition of (e2), therefore,
due to induction hypothesis, we get:
\begin{equation*}
\nu' _{2}\geq \frac{6}{7}n' \text{ and }n'\geq \frac{7}{2}n'_0.
\end{equation*}%
(c1), (c2) and (c4) imply that%
\begin{gather*}
n_0=n'_0+2;
\\
n =n' +k(d)+k(e)+4%
\text{,} \\
\nu _{2}\geq \nu' _{2}+k(d)+k(e)+3\text{.}
\end{gather*}%
Since $k(d)\geq 2$, $k(e)\geq 1$ we have
\begin{gather*}
\frac{k(d)+k(e)+3}{k(d)+k(e)+4}\geq \frac{6}{7}\text{, and} \\
\frac{k(d)+k(e)+4}{2}\geq \frac{7}{2}
\end{gather*}%
therefore, due to proposition \ref{FractionInequality}, we get:
\begin{gather*}
\frac{\nu _{2}}{n }\geq \frac{\nu' _{2}+k(d)+k(e)+3}{n'+k(d)+k(e)+4}\geq \frac{%
6}{7}\text{, and} \\
\frac{n}{n_0}=\frac{%
n'+k(d)+k(e)+4}{n'_0+2}\geq \frac{7}{2}.
\end{gather*}

Case 3: $G_{0}$ contains exactly one loop $e$ and does not satisfy
the condition of the case 1.

Let $d$ be the edge adjacent to the edge $e$. Let $u_{0}$ be the vertex of $%
G_{0}$ that is incident to $d$ and $e$, and let $d=(u_{0},v_{0})$.
Let $h$
and $h^{\prime }$ be two edges that differ from $d$ and are incident to $%
v_{0}$. Finally, let $u$ and $v$ be the endpoints of $h$ and
$h^{\prime }$ that are not incident to $d$, respectively.

Subcase 3.1: $d=f$ and $u=v$.

Define a cubic pseudo-graph $G_{0}^{\prime }$ as follows (figure \ref%
{Reduction}):%
\begin{eqnarray*}
G_{0}^{\prime } &=&(G_{0}\backslash \{u,v_{0}\})\cup \{g\}\text{, where} \\
g &=&(u_{0},w)\text{,}
\end{eqnarray*}%
and consider the graph $G^{\prime }$ obtained from $G_{0}^{\prime }$ by $%
k^{\prime }(e^{\prime })$-subdividing each edge $e^{\prime }$ of $%
G_{0}^{\prime }$, where
\begin{equation*}
k^{\prime }(e^{\prime })=\left\{
\begin{array}{ll}
k(f)+k(h^{\prime \prime })-2 & \text{if }e^{\prime }=g\text{,} \\
k(e^{\prime }) & \text{otherwise.}%
\end{array}%
\right.
\end{equation*}%
Note that $n'_0 < n_0$ and $k^{\prime }(g)=k(f)+k(h^{\prime \prime
})-2\geq 3$
thus, due to induction hypothesis, we have:%
\begin{equation*}
\nu' _{2}\geq \frac{6}{7}n' \text{ and }n'\geq \frac{7}{2}n'_0.
\end{equation*}

On the other hand, it is not hard to see that
\begin{gather*}
n_0 =n'_0 +2;
\\
n =n'+k(h)+k(h^{\prime })+4\text{,} \\
\nu _{2}\geq \nu' _{2}+k(h)+k(h^{\prime })+3.
\end{gather*}%
As $k(h),k(h^{\prime })\geq 2$, we have
\begin{eqnarray*}
\frac{k(h)+k(h^{\prime })+3}{k(h)+k(h^{\prime })+4} &\geq &\frac{7}{8}>\frac{%
6}{7}\text{, and} \\
\frac{k(h)+k(h^{\prime })+4}{2} &\geq &4>\frac{7}{2},
\end{eqnarray*}%
therefore, due to proposition \ref{FractionInequality}, we get:%
\begin{gather*}
\frac{\nu _{2}}{n}\geq \frac{\nu' _{2}+k(h)+k(h^{\prime })+3}{n'
+k(h)+k(h^{\prime })+4}\geq \frac{6}{7}\text{, and} \\
\frac{n }{n_0 }=\frac{%
n'+k(h)+k(h^{\prime })+4}{n'_0+2}\geq \frac{7}{2}.
\end{gather*}

Subcase 3.2: $d\neq f$ or $u\neq v$.

Consider the cubic pseudo-graph $G_{0}^{\prime }$ obtained from
$G_{0}$ by cutting the loop $e$ and the graph $G^{\prime }$ obtained
from a cubic pseudo-graph $G_{0}^{\prime }$ by $k^{\prime
}(e^{\prime })$-subdividing each edge $e^{\prime }$ of
$G_{0}^{\prime }$, where the mapping $k^{\prime }$ is defined
according to (\ref{KPrimeDefinition}). Note that $n'_0 <n_0$.

Let us show that $G_{0}^{\prime }$ and its corresponding graph
$G^{\prime }$ satisfy
\begin{equation}
\nu' _{2}\geq \frac{6}{7}n' \text{ and }n'\geq \frac{7}{2}n'_0.
\label{BasicInequalities}
\end{equation}%
Note that if $f\in E_{0}^{\prime },$ then, since $n'_0 < n_0$ and $%
k^{\prime }(f)=k(f)\geq 3$, (\ref{BasicInequalities}) follows
directly from the induction hypothesis. So, let us assume, that
$f\notin E_{0}^{\prime }$. If $d=f$ then $G_{0}^{\prime }$ does not
contain a loop as $u\neq v$. Thus (\ref{BasicInequalities}) follows
from (a1) and (a2). Thus, we may also assume that $d\neq f$. As
$f\notin E_{0}^{\prime }$, we deduce that $f\in \{h,h^{\prime }\}$.
As $G_{0}$ does not satisfy the condition of the case 1, we have
$u\neq v$ and $G_{0}^{\prime }$ does not contain a loop. Thus
(\ref{BasicInequalities}) again follows from (a1) and (a2).

Now, (c1), (c2) and (c4) imply that%
\begin{gather*}
n_0=n'_0+2;
\\
n =n'+k(d)+k(e)+4%
\text{,} \\
\nu _{2}\geq \nu' _{2}+k(d)+k(e)+3.
\end{gather*}%
Since $k(d)\geq 2$, $k(e)\geq 1$, we have
\begin{gather*}
\frac{k(d)+k(e)+3}{k(d)+k(e)+4}\geq \frac{6}{7}\text{, and} \\
\frac{k(d)+k(e)+4}{2}\geq \frac{7}{2}
\end{gather*}%
therefore, due to (\ref{BasicInequalities}) and proposition \ref%
{FractionInequality}, we get:
\begin{gather*}
\frac{\nu _{2}}{n}\geq \frac{\nu' _{2}+k(d)+k(e)+3}{n'+k(d)+k(e)+4}\geq \frac{%
6}{7}\text{, and} \\
\frac{n }{n_0}=\frac{%
n'+k(d)+k(e)+4}{n'_0+2}\geq \frac{7}{2}.
\end{gather*}

(f) Note that if $G_{0}$ satisfies at least one of the conditions of
(a), (e1), (e2), then, taking into account the inequality $2\nu
_{1}\geq \nu _{2}$, we get:
\begin{equation*}
\nu _{1}\geq \frac{\nu _{2}}{2}\geq \frac{1}{2}\cdot \frac{6}{7}%
n =\frac{3}{7}n,
\end{equation*}%
thus, without loss of generality, we may assume that $G_{0}$
satisfies none of the conditions of (a), (e1), (e2), hence $G_{0}$
contains at least one loop, and for each loop $e$ and for each edge
$f$ of $G_{0}$, that is not a loop, we have: $k(e)=1$ and $k(f)=2$.
For these cubic pseudo-graphs, we will prove the inequality (f) by
induction on $n_0$. If $n_0=2,$ then $G_{0}$ is the cubic
pseudo-graph from the figure \ref{TrivialCase} and, as $k(e)=k(g)=1$ and $%
k(f)=2$, $G$ contains a perfect matching, thus%
\begin{equation*}
\nu _{1}=\frac{1}{2}n >\frac{3}{7}n.
\end{equation*}%
Now, by induction, assume that for every graph $G^{\prime }$
obtained from a cubic pseudo-graph $G_{0}^{\prime }$ ($n'_0 <n_0$)
by $k^{\prime }(e^{\prime }) $-subdividing each edge $e^{\prime }$
of $G_{0}^{\prime }$, we have
\begin{equation*}
\nu' _{1}\geq \frac{3}{7}n',
\end{equation*}%
and consider the cubic pseudo-graph $G_{0}$ ($n_0 \geq 4$) and its
corresponding graph $G$.

Let $e$ be a loop of $G_{0}$, and consider a cubic pseudo-graph $%
G_{0}^{\prime }$ obtained from $G_{0}$ by cutting the loop $e$ and a graph $%
G^{\prime }$ obtained from $G_{0}^{\prime }$ by $k^{\prime }(d^{\prime })$%
-subdividing each edge $d^{\prime }$ of $G_{0}^{\prime }$, where the
mapping $k^{\prime }$ is defined according to
(\ref{KPrimeDefinition}). As $n'_0 < n_0$, due to induction
hypothesis, we have
\begin{equation*}
\nu' _{1}\geq \frac{3}{7}n'
\end{equation*}%
(c2) and (c3) imply that%
\begin{gather*}
n=n'+7,
\\
\nu _{1}\geq \nu' _{1}+3.
\end{gather*}%
Due to proposition \ref{FractionInequality}, we get:
\begin{equation*}
\frac{\nu _{1}}{n}\geq \frac{\nu' _{1}+3}{n'+7}\geq \frac{3}{7}.
\end{equation*}

(g) Let $G_{0}$ be the connected cubic pseudo-graph corresponding to
$G$ and let $\bar{G}_{0}$ be the tree obtained from $G_{0}$ by
removing its loops (see the definition of the class $\mathfrak{M}$).
Assume $k$ and $k^{\prime
} $ to be the numbers of internal (non-pendant) and pendant vertices of $%
\bar{G}_{0}$. Clearly, $k+k^{\prime }=\bar{n}_{0}=n_0$. On the other
hand,
\begin{equation*}
\bar{m}_{0}=m_0 -k^{\prime }=\frac{3}{2}(k+k^{\prime })-k^{\prime
}\text{.}
\end{equation*}%
Since $\bar{m}_{0}=\bar{n}_{0}-1$, we get
\begin{equation*}
k+k^{\prime }-1=\frac{3}{2}(k+k^{\prime })-k^{\prime }
\end{equation*}%
or
\begin{equation*}
k^{\prime }=k+2\text{.}
\end{equation*}%
We prove the inequality by induction on $k$. Note that if $k=0$ then
$G_{0}$ is the cubic pseudo-graph from the figure \ref{TrivialCase},
therefore
\begin{equation*}
\frac{\nu _{1}}{n}=\frac{3}{6}=\frac{1}{2}>%
\frac{6}{13}\text{.}
\end{equation*}%
On the other hand, if $k=1,$ then $G_{0}$ is the cubic pseudo-graph
shown on the figure \ref{Case k1}, thus
\begin{equation*}
\frac{\nu _{1}}{n}=\frac{6}{13}.
\end{equation*}

\begin{figure}[h]
\begin{center}
\includegraphics[height=15pc]{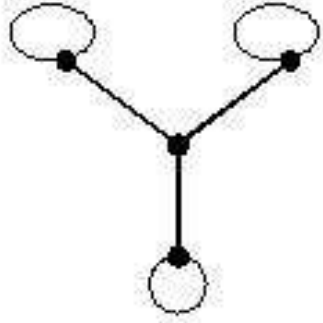}\\
\caption{The case $k=1$}\label{Case k1}
\end{center}
\end{figure}

Now, by induction, assume that for every graph $G^{\prime }\in
\mathfrak{M},$ we have $\nu' _{1}\geq \frac{6}{13}n'$, if the tree
$\bar{G}_{0}^{\prime }$ contains less than $k$ internal vertices,
and let us consider the graph $G\in \mathfrak{M}$ the corresponding
tree $\bar{G}_{0}$ of which contains $k$($k\geq 2$) internal
vertices. We need to consider two cases:

Case 1: There is $U=\{u_{1},...,u_{7}\}\subseteq \bar{V}_{0}$ such that $%
d_{\bar{G}_{0}}(u_{i})=1$, $1\leq i\leq 4$ and the subtree of
$\bar{G}_{0}$
induced by $U$ is the tree shown on the figure \ref{Case of two branches}.%

\begin{figure}[h]
\begin{center}
\includegraphics[height=15pc]{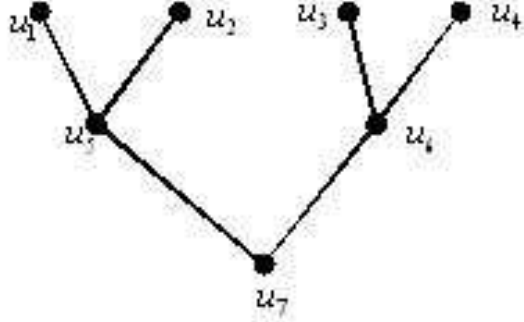}\\
\caption{The case of two branches}\label{Case of two branches}
\end{center}
\end{figure}

Let $\bar{G}_{0}^{\prime }$ be the tree $\bar{G}_{0}\backslash
\{u_{1},...,u_{6}\}$ and let $G_{0}^{\prime }$ be the cubic
pseudo-graph obtained from $G_{0}$ by removing the vertices
$u_{1},...,u_{6}$ and adding a new loop incident to the vertex
$u_{7}$. Note that $\bar{G}_{0}^{\prime }$ contains less than $k$
internal vertices, thus for the graph $G^{\prime }\in \mathfrak{M}$
corresponding to $\bar{G}_{0}^{\prime }$, we have
\begin{equation}
\frac{\nu' _{1}}{n'}\geq \frac{6}{13}. \label{InductionBound_6_13}
\end{equation}%
On the other hand, since
\begin{gather*}
n =n'-1+(6+16)=n'+21\text{,} \\
\nu _{1}\geq \nu' _{1}-1+11=\nu' _{1}+10,
\end{gather*}%
due to (\ref{InductionBound_6_13}) and proposition
\ref{FractionInequality}, we get:
\begin{equation*}
\frac{\nu _{1}}{n }\geq \frac{\nu' _{1}+10}{n'+21}\geq \frac{6}{13}
\end{equation*}%
since
\begin{equation*}
\frac{10}{21}>\frac{6}{13}.
\end{equation*}

Case 2: There is $U=\{u_{1},...,u_{6}\}\subseteq \bar{V}_{0}$ such that $%
d_{\bar{G}_{0}}(u_{1})=d_{\bar{G}_{0}}(u_{2})=d_{\bar{G}_{0}}(u_{5})=1$
and
the subtree of $\bar{G}_{0}$ induced by $U$ is the tree shown on the figure %
\ref{Case of branch and a leave}.

\begin{figure}[h]
\begin{center}
\includegraphics[height=15pc]{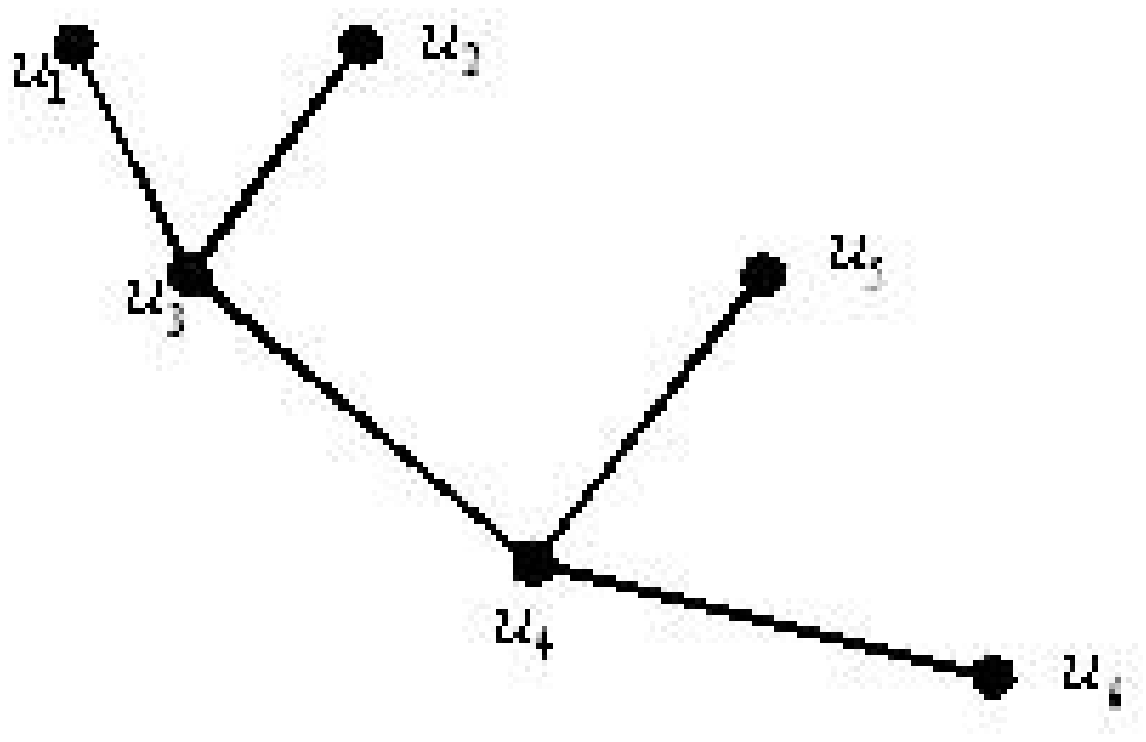}\\
\caption{The case of a branch and a leave}\label{Case of branch and
a leave}
\end{center}
\end{figure}

Let $\bar{G}_{0}^{\prime }$ be the tree $(\bar{G}_{0}\backslash
\{u_{1},...,u_{4}\})\cup \{(u_{5},u_{6})\}$ and let $G_{0}^{\prime
}$ be the
cubic pseudo-graph obtained from $G_{0}$ by removing the vertices $%
u_{1},u_{2},u_{3}$ and $u_{4}$ and adding the edge $(u_{5},u_{6})$.
Note that $\bar{G}_{0}^{\prime }$ contains less than $k$ internal
vertices, thus
for the graph $G^{\prime }\in \mathfrak{M}$ corresponding to $\bar{G}%
_{0}^{\prime }$, we have
\begin{equation}
\frac{\nu' _{1}}{n'}\geq \frac{6}{13}.
\label{InductionBound_6_13_Second}
\end{equation}%
On the other hand, since
\begin{gather*}
n =n'-2+14=n'+12\text{,} \\
\nu _{1}\geq \nu' _{1}-1+8=\nu' _{1}+7,
\end{gather*}%
due to (\ref{InductionBound_6_13}) and proposition
\ref{FractionInequality}, we get:
\begin{equation*}
\frac{\nu _{1}}{n}\geq \frac{\nu' _{1}+7}{n'+12}\geq \frac{6}{13},
\end{equation*}%
since
\begin{equation*}
\frac{7}{12}>\frac{6}{13}.
\end{equation*}%
To complete the proof of the inequality, let us note that, since the tree $%
\bar{G}_{0}$ contains $k$, ($k\geq 2$) internal vertices,
$\bar{G}_{0}$ satisfies at least one of the conditions of case 1 and
case 2.

(h) (c1) and (c2) imply that
\begin{gather*}
n_0=n'_0 +2;
\\
n =n' +k(f)+k(e)+4%
\text{.}
\end{gather*}
Since $k(f)\geq 2$, $k(e)\geq 1$ we have
\begin{equation*}
\frac{k(f)+k(e)+4}{2}\geq \frac{7}{2}\text{,}
\end{equation*}%
thus, due to proposition \ref{FractionInequality}, we get:%
\begin{equation*}
\frac{n}{n_0}=\frac{%
n'+k(f)+k(e)+4}{n'_0 +2}\geq \frac{7}{2}\text{.}
\end{equation*}

(i) Note that as $n <\frac{7}{2}n_0$ due to (e1) and (e2), for every
edge $e$
of $G_{0}$ we have%
\begin{equation*}
k(e)=\left\{
\begin{array}{ll}
1, & \text{if }e\text{ is a loop,} \\
2, & \text{otherwise.}%
\end{array}%
\right.
\end{equation*}%
Let us show that $G\in \mathfrak{M}$. Consider a maximal (with
respect to
the operation of cutting loops) sequence of cubic pseudo-graphs $%
G_{0}^{(0)},G_{0}^{(1)},...,G_{0}^{(n)}$, where $G_{0}^{(0)}=G_{0}$, and $%
G_{0}^{(i+1)}$ is obtained from $G_{0}^{(i)}$ by cutting a loop $e_{i}$ of $%
G_{0}^{(i)}$,$i=0,...,n-1$. Note that proposition
\ref{ConnectedInvariance} implies that for $i=1,...,n$ the graph
$G_{0}^{(i)}$ is connected.

Consider the sequence of graphs $G^{(0)},G^{(1)},...,G^{(n)}$, where $%
G^{(0)}=G$, and for $i=1,...,n$ the graph $G_{i}$ is obtained from $%
G_{0}^{(i)}$ by $k_{i}(d_{i})$-subdividing each edge $d_{i}$ of $G_{0}^{(i)}$%
, where the mapping $k_{i}$ is defined from $k_{i-1}$ according to (\ref%
{KPrimeDefinition}) and $k_{0}=k$. As the sequence $%
G_{0}^{(0)},G_{0}^{(1)},...,G_{0}^{(n)}$ is maximal, the operation
of cutting the loops is not applicable to $G_{0}^{(n)}$, thus due to
remark \ref{SuccessiveCut}, $G_{0}^{(n)}$ is either the trivial
cubic pseudo-graph from the figure \ref{TrivialCase} or a connected
graph (i.e. a connected pseudo-graph without loops). On the other
hand, (h) implies that for $i=1,...,n$, we have
\begin{equation}
n^{(i)}<\frac{7}{2}n_{0}^{(i)}\label{InequalityLIST}
\end{equation}%
thus, taking into account (a2), we deduce that $G_{0}^{(n)}$ is the
trivial cubic pseudo-graph from the figure \ref{TrivialCase}.

Note that for the proof of $G\in \mathfrak{M}$, it suffices to show
that if we remove all loops of $G_{0}$ then we will get a tree,
which is equivalent to proving that $G_{0}$ does not contain a
cycle. Suppose that $G_{0}$ contains a cycle. As $G_{0}^{(n)}$,
which is the pseudo-graph from the
figure \ref{TrivialCase}, does not contain a cycle, we imply that there is $%
j,1\leq j\leq n-1$ such that $G_{0}^{(j)}$ contains a cycle and $%
G_{0}^{(j+1)}$ does not. Proposition \ref{CycleInvariance} implies
that the
loop $e_{j}$ of $G_{0}^{(j)}$, whose cut led to the cubic pseudo-graph $%
G_{0}^{(j+1)}$ is adjacent to an edge $f_{j}$ which, in its turn, is
adjacent to two edges $h_{j}$ and $h_{j}^{\prime }$ that form the
only cycle of $G_{0}^{(j)}$.

As the edges $h_{j}$ and $h_{j}^{\prime }$ form a cycle of
$G_{0}^{(j)}$,
the cut of the loop $e_{j}$ leads to a loop $g_{j+1}$ of $%
G_{0}^{(j+1)}$ (see the definition of the operation of the cut of
loops). Due to (\ref{KPrimeDefinition}), we have
\begin{equation*}
k_{j+1}(g_{j+1})=k_{j}(h_{j})+k_{j}(h_{j}^{\prime })-2=2
\end{equation*}%
thus, due to (e1), we have%
\begin{equation*}
n^{(j+1)}\geq \frac{7}{2}n_{0}^{(j+1)}
\end{equation*}%
contradicting (\ref{InequalityLIST}).The proof of the lemma is
completed.
\end{proof}

\section{The main results}

We are ready to prove the first result of the paper. The basic idea
of the proof of this theorem can be roughly described as follows:
proving a lower bound for the main parameters of a cubic graph $G$
is just proving a bound for the graph $G\backslash F$ obtained by
removing a maximum matching $F$ of $G$. Next, according to lemma
\ref{Max Matching 2-3}, there is a maximum matching of a cubic graph
such that its removal leaves a graph, in which each degree is either
two or three. Moreover, the vertices of degree three are not placed
very closed. This allows us to consider this graph as a subdivision
of a cubic pseudo-graph, in which each edge is subdivided
sufficiently many times. The word "sufficiently" here should be
understood as big enough to allow us to apply the main results of
the lemma \ref{PseudoGraphSubdivision}. Next, by considering the
connected components of $G\backslash F$, we divide them into two or
three groups. For each of this groups, thanks to lemma
\ref{PseudoGraphSubdivision}, we find a bound for our parameters.
Then, due to proposition \ref{LinearInequality}, we not only
estimate the total contribution of the connected components to the
main parameters, but also keep this estimations big enough, which
allows us to get the main results of the theorem.

\begin{theorem}
\label{MainTheoremCubics}Let $G$ be a cubic graph. Then:%
\begin{equation*}
\nu _{1}\geq \frac{2}{5}n ,\nu _{2}\geq \frac{%
4}{5}n ,\nu _{3}\geq \frac{7}{6}n.
\end{equation*}
\end{theorem}

\begin{proof}
In \cite{Takao} it is shown that every odd regular graph $G$
contains a matching of size at least $\left\lceil \frac{(r^{2}-r-1)n
-(r-1)}{r(3r-5)}\right\rceil $, where $r$ is the degree of
vertices of $G$. Particularly, for a cubic graph $G$ we have:%
\begin{equation*}
\nu _{1}\geq \left\lceil \frac{5n -2}{12}%
\right\rceil \geq \frac{2}{5}n.
\end{equation*}

Now, let us show that the other two inequalities are also true. Let
$F$ be a
maximum matching of $G$ such that the unsaturated vertices (with respect to $%
F$) do not have a common neighbour (see lemma \ref{Max Matching 2-3}). Let $%
\varepsilon $ be a rational number such that $\varepsilon \in \lbrack 0,%
\frac{1}{10}]$ and
\begin{equation*}
\nu _{1}=\left\vert F\right\vert =(\frac{2}{5}+\varepsilon )n.
\end{equation*}%
Note that to complete the proof, it suffices to show that
\begin{equation*}
\nu _{1}(G\backslash F)\geq (\frac{2}{5}-\varepsilon )n,\nu
_{2}(G\backslash F)\geq (\frac{23}{30}-\varepsilon )n.
\end{equation*}%
Consider the graph $G\backslash F$. Clearly,
\begin{equation*}
2=\delta (G\backslash F)\leq \Delta (G\backslash F)\leq 3.
\end{equation*}

Let $x$ and $y$ be the numbers of vertices of $G\backslash F$ with
degree
two and three, respectively. Clearly,%
\begin{equation*}
\left\{
\begin{array}{l}
x+y=\left\vert V(G\backslash F)\right\vert =n , \\
2x+3y=2m -2\left\vert F\right\vert =3n -(\frac{4}{5}+2\varepsilon )n =(%
\frac{11}{5}-2\varepsilon )n ,%
\end{array}%
\right.
\end{equation*}

which implies that
\begin{equation*}
x=(\frac{4}{5}+2\varepsilon )n ,y=(\frac{1}{5}%
-2\varepsilon )n.
\end{equation*}%
Let $G_{1},...,G_{r}$ be the connected components of $G\backslash
F$. For a vertex $v_{i}\in V_{i}$, $1\leq i\leq r$ define:
\begin{equation*}
\nu _{1}(v_{i})=\frac{\nu _{1i}}{n_i },\nu _{2}(v_{i})=\frac{\nu
_{2i}}{n_i}.
\end{equation*}

Note that
\begin{eqnarray}
\frac{\nu _{1}(G\backslash F)}{\left\vert V(G\backslash F)\right\vert } &=&%
\frac{\nu _{1}(G\backslash F)}{n}=\frac{\nu _{1,1}+...+\nu
_{1,r}}{n_1+...+n_r}=  \notag \\
&=&\frac{n_1 \cdot \frac{\nu _{1,1}}{%
n_1}+...+n_r \cdot \frac{\nu _{1,r}}{n_r }}{n_1+...+n_r}=  \notag \\
&=&\frac{n_1 \cdot \nu _{1}(v_{1})+...+n_r \cdot \nu
_{1}(v_{r})}{n_1+...+n_r}, \label{Nhu1Weight}
\end{eqnarray}%
and similarly
\begin{equation}
\frac{\nu _{2}(G\backslash F)}{\left\vert V(G\backslash F)\right\vert }=%
\frac{n_1 \cdot \nu _{2}(v_{1})+...+n_r \cdot \nu _{2}(v_{r})}{n_1
+...+n_r}  \label{Nhu2Weight}
\end{equation}%
where $v_{1},...,v_{r}$ are vertices of $G\backslash F$ with
$v_{i}\in V(G_{i})$, $1\leq i\leq r$.

By the choice of $F$, we have that for $i=1,...,r$ $G_{i}$ is

\begin{description}
\item[(a)] either a cycle,

\item[(b)] or a connected graph, with $\delta_{i}=2,\Delta_{i}=3$
which does not contain two vertices of degree three that are
adjacent or share a neighbour.
\end{description}

Note that if $G_{i}$ is of type (b), then there is a cubic pseudo-graph $%
G_{i}^{0}$ such that $G_{i}$ can be obtained from $G_{i}^{0}$ by $k(e)$%
-subdividing each edge $e$ of $G_{i}^{0}$ (proposition \ref{CubicPseudoGraph}%
). Of course, if $e$ is not a loop then $k(e)\geq 2$.

Let $a,b,c$ be the numbers of vertices of $G\backslash F$ that lie
on its connected components $G_{1},...,G_{r}$, which are cycles,
graphs of type (b) that are from the class $\mathfrak{M}$, graphs of
type (b) which are not from the class $\mathfrak{M}$, respectively.

It is clear that if $v_{a}$ is a vertex of $G\backslash F$ lying on
a cycle of length $l$ then
\begin{equation*}
\nu _{1}(v_{a})=\frac{\left[ \frac{l}{2}\right] }{l}\geq
\frac{1}{3}\text{.}
\end{equation*}%
If $v_{b}$ is a vertex of $G\backslash F$ lying on a connected component $%
G_{b}$ of $G\backslash F$ which is from the class $\mathfrak{M}$,
then (g) of lemma \ref{PseudoGraphSubdivision} implies that
\begin{equation*}
\nu _{1}(v_{b})=\frac{\nu _{1b}}{n_b}\geq \frac{6}{13}.
\end{equation*}%
If $v_{c}$ is a vertex of $G\backslash F$ lying on a connected component $%
G_{c}$ of $G\backslash F$ which is of type (b) and does not belong
to the class $\mathfrak{M}$, then (f) of lemma
\ref{PseudoGraphSubdivision} implies
that%
\begin{equation*}
\nu _{1}(v_{c})=\frac{\nu _{1c}}{n_c}\geq \frac{3}{7}\text{.}
\end{equation*}

Let $k_{b}$ and $k_{c}$ be the number of vertices of $G\backslash F$
with degree three that lie on connected components
$G_{1},...,G_{r}$, which are graphs from the class $\mathfrak{M}$ or
are graphs of type (b), which are not from the class $\mathfrak{M}$,
respectively. Clearly,
\begin{equation}
k_{b}+k_{c}=y=(\frac{1}{5}-2\varepsilon )n. \label{ThreeVertices}
\end{equation}%
(d2) of lemma \ref{PseudoGraphSubdivision} implies that%
\begin{equation*}
b\geq 3k_{b}\text{.}
\end{equation*}%
(i) of lemma \ref{PseudoGraphSubdivision} implies that
\begin{equation*}
c\geq \frac{7}{2}k_{c}\text{.}
\end{equation*}%
Thus, due to (\ref{Nhu1Weight})%
\begin{equation*}
\frac{\nu _{1}(G\backslash F)}{\left\vert V(G\backslash
F)\right\vert }\geq
\frac{\frac{1}{3}a+\frac{6}{13}b+\frac{3}{7}c}{n}%
\text{.}
\end{equation*}%
As $a+b+c=n$ we get: $a\leq n -3k_{b}-\frac{7}{2}k_{c}$. Since $\frac{1}{3}<\frac{3}{7}<%
\frac{6}{13}$, due to proposition \ref{LinearInequality}, we have:%
\begin{equation*}
\frac{1}{3}a+\frac{6}{13}b+\frac{3}{7}c\geq \frac{1}{3}(n-3k_{b}-\frac{7}{2}k_{c})+\frac{6}{13}\cdot 3k_{b}+\frac{3}{7%
}\cdot \frac{7}{2}k_{c}
\end{equation*}%
and therefore%
\begin{eqnarray*}
\frac{\nu _{1}(G\backslash F)}{\left\vert V(G\backslash
F)\right\vert }
&\geq &\frac{\frac{1}{3}(n -3k_{b}-\frac{7}{2}%
k_{c})+\frac{6}{13}\cdot 3k_{b}+\frac{3}{7}\cdot \frac{7}{2}k_{c}}{%
n }= \\
&=&\frac{\frac{1}{3}n +\frac{5}{13}k_{b}+\frac{1}{3}%
k_{c}}{n }=\frac{1}{3}+\frac{1}{3}\frac{k_{b}+k_{c}%
}{n }+\frac{2}{39}\frac{k_{b}}{n }
\end{eqnarray*}%
(\ref{ThreeVertices}) implies that
\begin{eqnarray*}
\frac{\nu _{1}(G\backslash F)}{n }
&\geq &\frac{1}{3}+\frac{1}{3}(\frac{1}{5}-2\varepsilon )+\frac{2}{39}\frac{%
k_{b}}{n }=\frac{2}{5}-\frac{2}{3}\varepsilon +%
\frac{2}{39}\frac{k_{b}}{n }= \\
&=&(\frac{2}{5}-\varepsilon )+\frac{\varepsilon }{3}+\frac{2}{39}\frac{k_{b}%
}{n}\geq \frac{2}{5}-\varepsilon
\end{eqnarray*}%
which is equivalent to%
\begin{equation*}
\nu _{1}(G\backslash F)\geq (\frac{2}{5}-\varepsilon )\left\vert
V(G\backslash F)\right\vert =(\frac{2}{5}-\varepsilon )n.
\end{equation*}%
Note that if $\nu _{2}=\frac{4}{5}n$, then
$\varepsilon =0,k_{b}=0$, which means that $\nu _{1}=\frac{2}{5}%
n$ and among the components $G_{1},...,G_{r}$ there are no
representatives of the class $\mathfrak{M}$.

Now, let us turn to the proof of the inequality $\nu
_{2}(G\backslash F)\geq (\frac{23}{30}-\varepsilon )n$.

Let $A,B$ be the numbers of vertices of $G\backslash F$ that lie on
its connected components $G_{1},...,G_{r}$, which are cycles and
graphs of type (b), respectively. It is clear that if $v_{A}$ is a
vertex of $G\backslash F$ lying on a cycle of the length $l$ then
\begin{equation*}
\nu _{2}(v_{A})=\frac{2\left[ \frac{l}{2}\right] }{l}\geq
\frac{2}{3}\text{.}
\end{equation*}%
If $v_{B}$ is a vertex of $G\backslash F$ lying on a connected component $%
G_{B}$ of $G\backslash F$ which is of type (b), then (d1) of lemma \ref%
{PseudoGraphSubdivision} implies that
\begin{equation*}
\nu _{2}(v_{B})=\frac{\nu _{2B}}{n_B}\geq \frac{5}{6}.
\end{equation*}%
As the number of vertices of $G\backslash F$ which are of degree three is $%
y=(\frac{1}{5}-2\varepsilon )n $, (d2) of lemma \ref%
{PseudoGraphSubdivision} implies that
\begin{equation}
B\geq 3y=(\frac{3}{5}-6\varepsilon )n. \label{Bbound}
\end{equation}%
Thus, due to (\ref{Nhu2Weight})
\begin{equation*}
\frac{\nu _{2}(G\backslash F)}{\left\vert V(G\backslash
F)\right\vert }\geq
\frac{\frac{2}{3}A+\frac{5}{6}B}{\left\vert V(G\backslash F)\right\vert }%
\text{.}
\end{equation*}%
As $A+B=n$, (\ref{Bbound}) implies that
\begin{equation*}
A\leq n -3y=n -(\frac{3}{5}%
-6\varepsilon )n =(\frac{2}{5}+6\varepsilon )n.
\end{equation*}%
Since $\frac{2}{3}<\frac{5}{6}$, due to proposition
\ref{LinearInequality}, we get
\begin{equation*}
\frac{2}{3}A+\frac{5}{6}B\geq \frac{2}{3}(\frac{2}{5}+6\varepsilon
)n +\frac{5}{6}(\frac{3}{5}-6\varepsilon )n
\end{equation*}%
and therefore%
\begin{equation*}
\frac{\nu _{2}(G\backslash F)}{\left\vert V(G\backslash
F)\right\vert }\geq \frac{(\frac{23}{30}-\varepsilon )n}{\left\vert
V(G\backslash F)\right\vert }=(\frac{23}{30}-\varepsilon )\text{,}
\end{equation*}%
which is equivalent to%
\begin{equation*}
\nu _{2}(G\backslash F)\geq (\frac{23}{30}-\varepsilon )\left\vert
V(G\backslash F)\right\vert =(\frac{23}{30}-\varepsilon )n.
\end{equation*}

The proof of the theorem is completed.
\end{proof}

\begin{remark}
There are graphs attaining bounds of the theorem
\ref{MainTheoremCubics}. The graph from figure \ref{Examples
attaining the bounds}a attains the first two bounds and the graph
from figure \ref{Examples attaining the bounds}b the last bound.
\end{remark}

\begin{center}
\begin{figure}[h]
\begin{center}
\includegraphics[height=10pc]{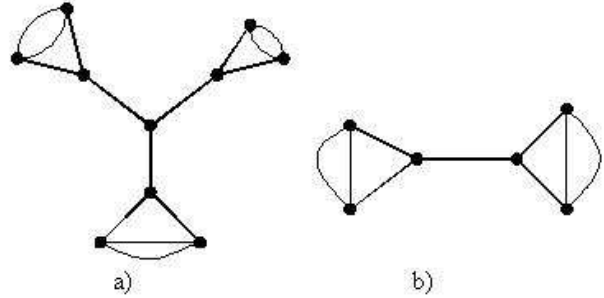}\\
\caption{Examples attaining the bounds of the theorem
\protect\ref{MainTheoremCubics}}\label{Examples attaining the
bounds}
\end{center}
\end{figure}
\end{center}

Recently, we managed to prove:

\begin{theorem}
For every cubic graph $G$%
\begin{equation*}
\nu _{2}+\nu _{3}\geq 2n.
\end{equation*}
\end{theorem}

Note that this implies that there is no graph attaining
all the bounds of theorem \ref{MainTheoremCubics} at the same time.

The methodology developed above allows us to prove the second result
of the paper, which is an inequality among our main parameters. To
prove it, again we reduce the inequality to another one considered
in the class of graphs, that are obtained from a cubic graph by
removing a matching of $G$. Note that this time matching need not to
be maximum, nevertheless, its removal keeps the vertices of degree
three "far enough". Next, by considering any of connected components
of this graph, we look at it as a subdivision of a cubic
pseudo-graph. This allows us to apply the results from the section
on system of cycles and paths, and find a suitable system, which not
only captures the essence of the inequality that we were trying to
prove, but also is very simple in its structure, and this allows us
to complete the proof.

\begin{theorem}
\label{ArithmeticalMean}For every cubic graph $G$ the following inequality holds:%
\begin{equation*}
\nu _{2}\leq \frac{n+2\nu _{3}}{4}.
\end{equation*}
\end{theorem}

\begin{proof}
Let $(H,H^{\prime })$ be a pair of edge-disjoint matchings of
$G$ with $\left\vert H\right\vert +\left\vert H^{\prime }\right\vert
=\nu _{2}$. Without loss of generality we may assume that $H$ is
maximal (not necessarily maximum). Let $G_{1},...,G_{k}$ be the
connected components of $G\backslash H$, $l_i=l(G_{i})$ be the
number of vertices of $G_{i}$ having
degree three, $1\leq i\leq k$, and let $l$ be the number of vertices of $%
G\backslash H$ having degree three. Note that%
\begin{equation*}
l=l_{1}+...+l_{k}=n -2\left\vert H\right\vert \text{.}
\end{equation*}

Let us show that for each $i$, $1\leq i\leq k$, the following
inequality is
true:%
\begin{equation}
\nu _{2i}\geq 2\nu _{1i}-\frac{l_i}{2}.
\label{InequalityInComponents}
\end{equation}

Note that, if $G_{i}$ is a cycle, then $l_{i}=0$ and $\nu _{2i}=2\nu
_{1i}$, thus (\ref{InequalityInComponents}) is true for the cycles.
Now, let us assume $G_{i}$ to contain a vertex of degree three. As
$H$ is a maximal matching, no two vertices of degree three are
adjacent in $G_{i}$. Proposition
\ref{CubicPseudoGraph} implies that there is a cubic pseudo-graph $G_{i}^{0}$ such that $G_{i}$ can be obtained from $%
G_{i}^{0}$ by $k(e)$-subdividing each edge $e$ of $G_{i}^{0}$ where $%
k(e)\geq 1$. Let $G_{i}^{\prime }$ be the graph obtained from
$G_{i}^{0}$ by $1$-subdividing each edge $e$ of $G_{i}^{0}$. Note
that $G_{i}^{\prime }$
contains $n_{i}^{0}$ vertices of degree three, $%
\frac{3n_{i}^{0}}{2}$ vertices of degree two and no two vertices of
the same degree are adjacent in $G_{i}^{\prime }$. Due to lemma
\ref{Bipartite 2->=3}, there is a system $\mathfrak{F}_{i}^{\prime
}$ of even cycles and paths of $G_{i}^{\prime }$ satisfying the
conditions
(1.2),(1.3) of the lemma \ref{Bipartite 2->=3} and containing $\frac{%
n_{i}^{0}}{2}$ paths (see (1.1) of the lemma \ref%
{Bipartite 2->=3}). (2) of lemma \ref{Bipartite 2->=3} implies that $%
\mathfrak{F}_{i}^{\prime }$ includes a maximum matching of
$G_{i}^{\prime }$.

Now, note that $G_{i}$ can be obtained from $G_{i}^{\prime }$ by a
sequence of $1$-subdivisions. Lemma \ref{SystemInSubdivision}
implies that there is a system $\mathfrak{F}_{i}$ of paths and even
cycles of $G_{i}$ satisfying the conditions (1)-(5) of the lemma
\ref{Bipartite 2->=3} and containing exactly $\frac{n_{i}^{0}}{2}$
paths!

Let $x$ be the number of paths from $\mathfrak{F}_{i}$ containing an
odd number
of edges. Note that since $x\leq \frac{n_{i}^{0}}{2%
}$, we have:
\begin{eqnarray*}
\nu _{2i} &\geq &\sum_{F\in \mathfrak{F}_{i}}\left\vert
E(F)\right\vert =2\sum_{F\in \mathfrak{F}_{i}}\nu _{1}(F)-x=2\nu
_{1i}-x\geq \\
&\geq &2\nu _{1i}-\frac{n_{i}^{0}}{2}=2\nu _{1i}-\frac{l_{i}}{2}.
\end{eqnarray*}%
Summing up the inequalities (\ref{InequalityInComponents}) from $1$
to $k$
we get:%
\begin{equation*}
\nu _{2}(G\backslash H)=\sum_{i=1}^{k}\nu _{2i}\geq
2\sum_{i=1}^{k}\nu _{1i}-\frac{\sum_{i=1}^{k}l_{i}}{%
2}=2\nu _{1}(G\backslash H)-\frac{l}{2}\text{.}
\end{equation*}%
Thus
\begin{equation*}
\nu _{3}\geq \left\vert H\right\vert +\nu _{2}(G\backslash H)\geq
\left\vert H\right\vert +2\nu _{1}(G\backslash
H)-\frac{l}{2}=\left\vert
H\right\vert +2\nu _{1}(G\backslash H)-\frac{n }{2}%
+\left\vert H\right\vert \text{.}
\end{equation*}%
Taking into account that
\begin{equation*}
\left\vert H\right\vert +\left\vert H^{\prime }\right\vert
=\left\vert H\right\vert +\nu _{1}(G\backslash H)=\nu _{2}
\end{equation*}%
we get:
\begin{equation*}
\nu _{3}\geq 2\nu _{2}-\frac{n}{2}
\end{equation*}%
or
\begin{equation*}
\nu _{2}\leq \frac{n+2\nu _{3}}{4}.
\end{equation*}%
The proof of the theorem \ref{ArithmeticalMean} is completed.
\end{proof}

\begin{acknowledgement}
We would like to thank our reviewers for their useful comments that helped us to improve the paper.
\end{acknowledgement}

\end{document}